%% file: main.tex
\documentclass[runningheads]{llncs}
\usepackage[utf8]{inputenc}
\usepackage{amsmath}
\usepackage{algorithm}
\usepackage[noend]{algpseudocode}
\usepackage{mathrsfs}
\usepackage{mathtools}
\usepackage{comment}
\usepackage{xspace}
\usepackage{bbm}
\usepackage{booktabs}
\usepackage{array}
\usepackage{multirow}
\usepackage{multicol}
\usepackage{xcolor}
\usepackage{graphicx}
\usepackage[firstpage]{draftwatermark}
\usepackage{hyperref}
\usepackage{float}
\usepackage[misc,geometry]{ifsym}

\usepackage{stmaryrd}
\usepackage{amssymb}

\usepackage[colorinlistoftodos]{todonotes}

\input{macros}

\SetWatermarkAngle{0}

\SetWatermarkText{} %

\ifextended
\else
\input{appendix-labels}
\fi
\begin{document}
\ifextended
\title{Synthesizing Trajectory Queries from Examples}
\else
\title{Synthesizing Trajectory Queries from Examples\thanks{Appendices are available in the technical report~\cite{quivrextended}.}}
\fi
\author{Stephen Mell\inst{1}\textsuperscript{(\Letter)}
\and
Favyen Bastani\inst{2}
\and
Steve Zdancewic\inst{1}
\and
Osbert Bastani\inst{1}
}

\authorrunning{S. Mell et al.}
\institute{University of Pennsylvania, Philadelphia PA 19104, USA \\
\email{\{sm1,stevez,obastani\}@cis.upenn.edu} \and
Allen Institute for AI, Seattle WA 98104, USA \\
\email{favyenb@allenai.org}}

\maketitle              %
\begin{abstract}
\input{abstract}

\end{abstract}

\setcounter{footnote}{0}
\input{body}

\bibliographystyle{splncs04}
\bibliography{main}

\ifextended
\clearpage
\appendix

\input{appendix}

\fi

\end{document}

%% file: macros.tex
\newcommand{\basepred}[1]{\langle #1 \rangle}

\newcommand{\seq}{\,;\,}
\newcommand{\den}[1]{\llbracket #1 \rrbracket}

\newcommand{\denq}[1]{\llbracket #1 \rrbracket^{q}}
\newcommand{\denmat}[1]{\llbracket #1 \rrbracket^{M}}

\newcommand{\boxlit}[2]{\left\lfloor #1, #2 \right\rceil}

\newcommand{\abs}[1]{\left| #1 \right|}
\newcommand{\toolname}{\textsc{Quivr}\xspace}

\newcommand{\brackparen}[1]{\mathopen{}\left( {#1}_{{}_{}}\,\negthickspace\right)\mathclose{}}
\newcommand{\brackangle}[1]{\mathopen{}\left\langle {#1}_{{}_{}}\,\negthickspace\right\rangle\mathclose{}}
\newcommand{\predbase}[1]{\brackangle{#1}}

\DeclarePairedDelimiter{\transstl}{\llparenthesis}{\rrparenthesis}

\newcommand{\para}[1]{\vspace{5pt}\noindent\emph{#1}}

\newcommand{\expredname}{\text{VelGt}}
\newcommand{\expred}[1]{\predbase{\expredname_{#1}}}

\newif\ifextended\extendedtrue

\ifextended
\newcommand{\refappendix}[1]{Appendix~\ref{#1}}

\newcommand{\refappendixtab}[1]{Table~\ref{#1}}
\else
\newcommand{\refappendix}[1]{Appendix~\ref{#1}}

\newcommand{\refappendixtab}[1]{Appendix~Table~\ref{#1}}
\fi

\makeatletter
\newcommand\notsotiny{\@setfontsize\notsotiny\@vipt\@viipt}
\makeatother

%% file: abstract.tex
Data scientists often need to write programs to process predictions of machine learning models, such as object detections and trajectories in video data. However, writing such queries can be challenging due to the fuzzy nature of real-world data; in particular, they often include real-valued parameters that must be tuned by hand. We propose a novel framework called \toolname that synthesizes trajectory queries matching a given set of examples. To efficiently synthesize parameters, we introduce a novel technique for pruning the parameter space and a novel quantitative semantics that makes this more efficient.
We evaluate \toolname on a benchmark of 17 tasks, including several from prior work, and show both that it can synthesize accurate queries for each task and that our optimizations substantially reduce synthesis time.

%% file: body.tex
\section{Introduction}

Over the past decade, deep neural networks (DNNs) have successfully solved challenging artificial intelligence problems~\cite{krizhevsky2017imagenet,silver2016mastering}. Abstractly, these models can be thought of as providing interfaces to real-world data---e.g., they can provide object classes~\cite{krizhevsky2017imagenet,he2016deep}, detections~\cite{ren2015faster,redmon2018yolov3}, and trajectories~\cite{bertinetto2016fully,wojke2017simple,bergmann2019tracking}. Then, these predictions are processed by programs, e.g., to identify driving patterns~\cite{bastaniskyquery}, events in TV broadcasts~\cite{fu2019rekall}, or animal behaviors~\cite{shah2020learning}.

However, writing such programs can be challenging since they must still account for the fuzziness of real data. To do so, these programs typically include real-valued parameters that need to be manually tuned by the user. For example, consider a query over car trajectories designed to identify instances where one car turns in front of another. This query must capture the shape of the trajectory of both the turning car and the car crossing the intersection. In addition, the user must select the appropriate maximum duration from the first car changing lanes to the second car crossing the intersection. Even an expert would require significant experimentation to determine good parameter values; in our experience, it can take up to an hour to tune the parameters for a single query.

We focus on programs that query databases of trajectories output by an object tracker~\cite{kang2019blazeit,fu2019rekall,kang2020jointly,kang2020model,bastani2020miris,bastani2020vaas,moll2020exsample,bastaniskyquery}. Given a video, the tracker predicts the positions of objects in each frame (e.g., cars, people, or mice), as well as associations between detections of the same object across successive frames. Applications often require subsequent analysis of these trajectories. For example, in autonomous driving, when a risky scenario is encountered, engineers typically search for additional examples of that driving pattern to improve their planner~\cite{sadigh2016information,sadigh2016planning,schmerling2018multimodal}---e.g., cars driving too close~\cite{wishart2020driving} or stopping in the middle of the road~\cite{bastani2021skyquery}. Object tracking has also been used to track robots~\cite{preiss2017crazyswarm,weinstein2018visual}, animals for behavioral analysis~\cite{tweed2002tracking,betke2007tracking,shah2020learning}, and basketball players for sports analytics~\cite{zhan2020learning,shah2020learning}.

We propose an algorithm for synthesizing queries over object trajectories given just a handful of input-output examples. A query takes as input a representation of a trajectory as a sequence of states (e.g., position, velocity, and acceleration) in successive frames of the video, and outputs whether the trajectory matches its semantics. Our query language is based on regular expressions---in particular, a query is a composition of a user-extensible set of predicates using the sequencing, conjunction, and iteration operators. For instance, trajectories might correspond to cars in a video; Figure~\ref{fig:example} shows a query for identifying cars turning at an intersection. As we discuss in Section~\ref{sec:related}, the full query language semantics is rich enough to subsume (variants of) Kleene algebras with tests (KAT)~\cite{kozen1997kleene} and signal temporal logic (STL)~\cite{stl}; however, such generality is seldom needed, so we use a pared-down query language that works well in practice.

Our algorithm performs enumerative search over the space of possible queries to identify ones that are consistent with the given examples.
A key challenge in our setting is that our predicates have real-valued parameters that must also be synthesized. Thus, our strategy enumerates \emph{sketches}, which are partial programs that only contain holes corresponding to real-valued parameters. For each sketch, we search over the space of real-valued parameters, while using an efficient pruning strategy to reduce the search space. At a high level, we use a quantitative semantics to directly compute ``boundary parameters'' at which a given example switches from being labeled positive to negative. Then, depending on the target label, we can prune the entire region of the search space on one side of these boundary parameters. We prove that this synthesis strategy comes with soundness and (partial) completeness guarantees.

We implement our approach in a system called \toolname.\footnote{\toolname stands for QUery Induction for Video tRajectories.} Our implementation focuses on videos from fixed-position cameras.
While our language and synthesis algorithm are general,
the predicates we design are tailored to specific settings. We evaluate \toolname on identifying driving patterns in traffic videos, including ones inspired by recent work on autonomous driving~\cite{sadigh2016planning,sadigh2016information,schmerling2018multimodal}, on behavior detection in a dataset of mouse trajectories~\cite{mabe22}, and on a synthetic task from the temporal logic synthesis literature~\cite{kong2017}. We demonstrate how both our parameter pruning strategies and our query evaluation optimizations lead to substantial reductions in the running time of our synthesizer.

In summary, our contributions are:
\begin{itemize}
\item A language for querying object trajectories (Section~\ref{sec:lang}) and an algorithm for synthesizing such queries from examples (Section~\ref{sec:algo}).
\item An efficient parameter pruning approach based on a novel quantitative semantics (Section~\ref{sec:algo}), yielding a $5.0 \times$ speedup over the state-of-the-art quantitative pruning technique from the temporal logic synthesis literature.
\item An implementation of our approach in \toolname, and an evaluation of \toolname on identifying driving behaviors in traffic camera video and mouse behaviors in a dataset of mouse trajectories (Section~\ref{sec:exp}), demonstrating substantially better accuracy than neural network baselines.
\end{itemize}

\begin{figure*}[t]
\centering
\begin{tabular}{cc}
\includegraphics[width=0.45\textwidth]{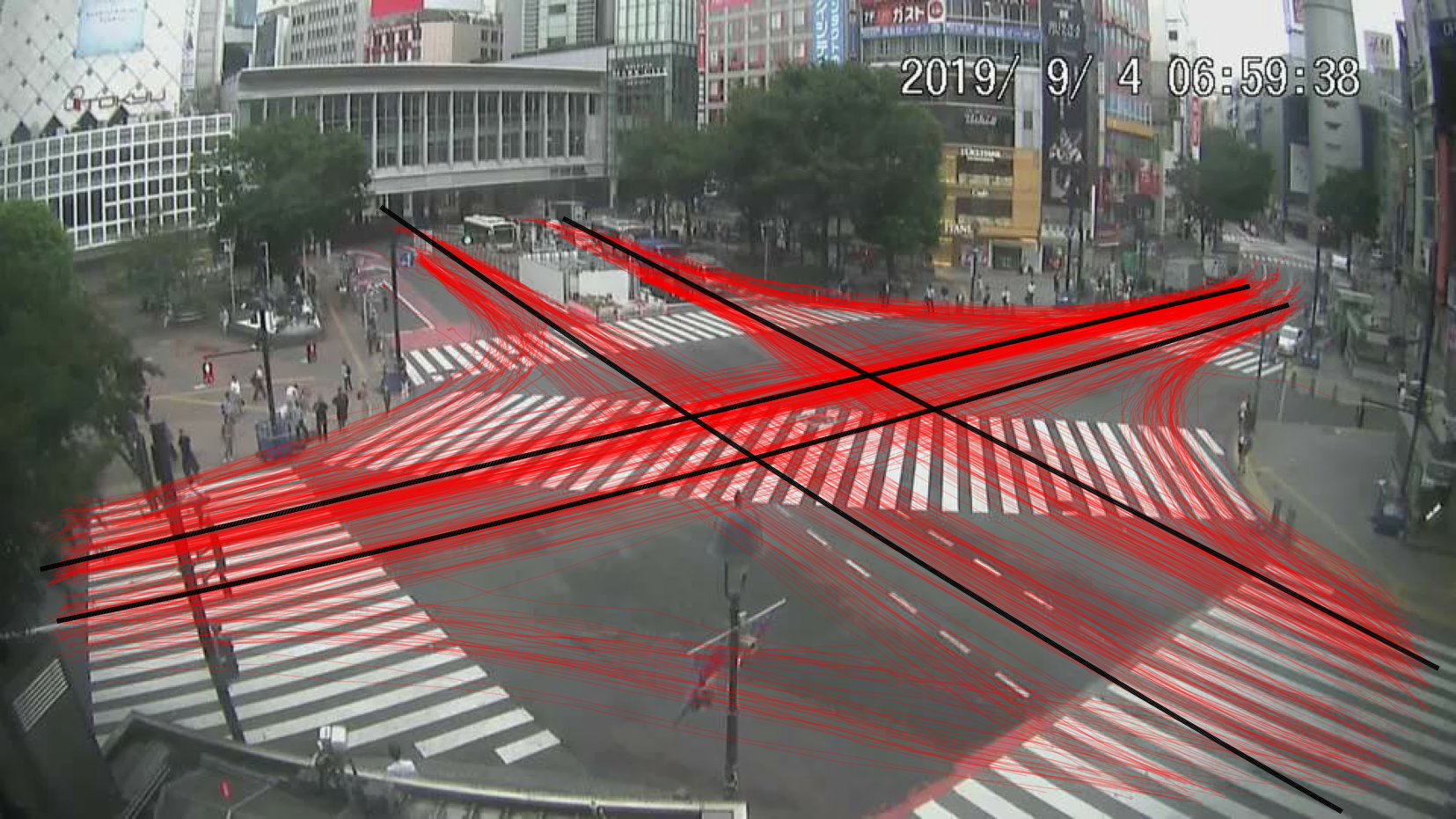} &
\includegraphics[width=0.45\textwidth]{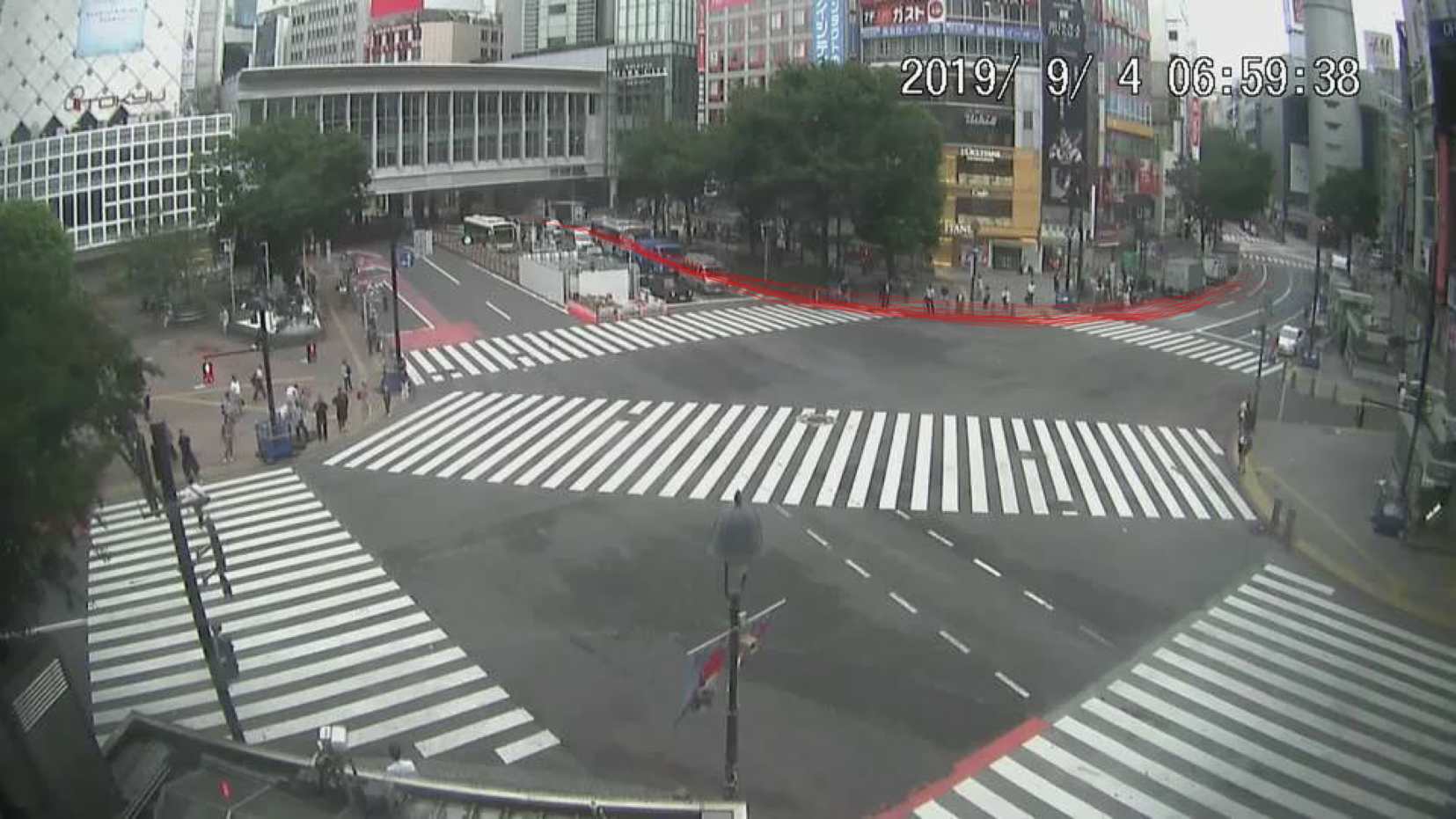} \\
(a) & (b)
\end{tabular} \vspace{5pt} \\
$\predbase{\textsf{InLane1}} \seq \predbase{\textsf{Any}} \seq \predbase{\textsf{InLane2}}$
\caption{(a) A video frame from a traffic camera, along with object trajectories (red) and manually annotated lanes (black). (b) The trajectories selected by the query (bottom),
which selects cars turning at the intersection.
}
\label{fig:example}
\end{figure*}

\section{Overview}
\label{sec:overview}

We consider a hypothetical scenario where an engineer is designing a control algorithm for an autonomous car and would like to identify certain driving patterns in video data. We show how they can use our framework to synthesize a query to identify car trajectories that exhibit a given behavior.

\para{Video data.}
Traffic cameras are a rich source of driving behaviors~\cite{robicquet2016learning,bock2019ind,bastaniskyquery};
one dataset used in our evaluation is YTStreams~\cite{bastani2020miris}, which includes video from several such cameras. 
Figure~\ref{fig:example} (a) shows a single frame from such a video;
we have used an object tracker~\cite{wojke2017simple} to identify all
car trajectories (in red).

\para{Predicates.}
\toolname assumes it is given a set of predicates that match portions of trajectories exhibiting behaviors of interest; during synthesis, it considers queries composed of these predicates. In Figure~\ref{fig:example} (a), the engineer has manually annotated the lanes of interest in this video (black), to specify four \textsf{InLaneK} predicates that select trajectories of cars driving in each lane $K$ visible in the video. Predicates may be configured by real-valued parameters. For example,
\begin{align*}
\basepred{\textsf{InLane1}}\wedge\basepred{\textsf{DispLt}_{\theta}}
\end{align*}
searches for trajectories where the car stays in lane 1 for a period of time
and the car has a displacement at most $\theta$ between the beginning and end of that period. Note that atomic predicates, like $\basepred{\textsf{DispLt}_{\theta}}$, can match multiple time-steps, whereas in formalisms like regular expressions and temporal logic, atomic predicates are over single time-steps.
A key feature of our framework is that the set of available predicates is highly extensible, and the user can provide their own. See Section~\ref{sec:setup} for the predicates we use in our evaluation.

\para{Synthesis.}
To specify a driving pattern, the engineer provides a small number of initial positive and negative examples of trajectories; then, \toolname synthesizes a query that correctly labels these examples.
In Figure~\ref{fig:example} (b), we show the result of executing the query shown, which is synthesized to identify left turns in the data. Often, there are multiple queries consistent with the initial examples. While it may be hard for users to sift through the video for positive examples, it is usually easy for them to label a given trajectory. Thus, to disambiguate, \toolname asks the user to label additional trajectories~\cite{roy2001toward,dasgupta2005analysis,ji2020question}.

\begin{figure}[t]
\centering
\includegraphics[width=0.55\textwidth]{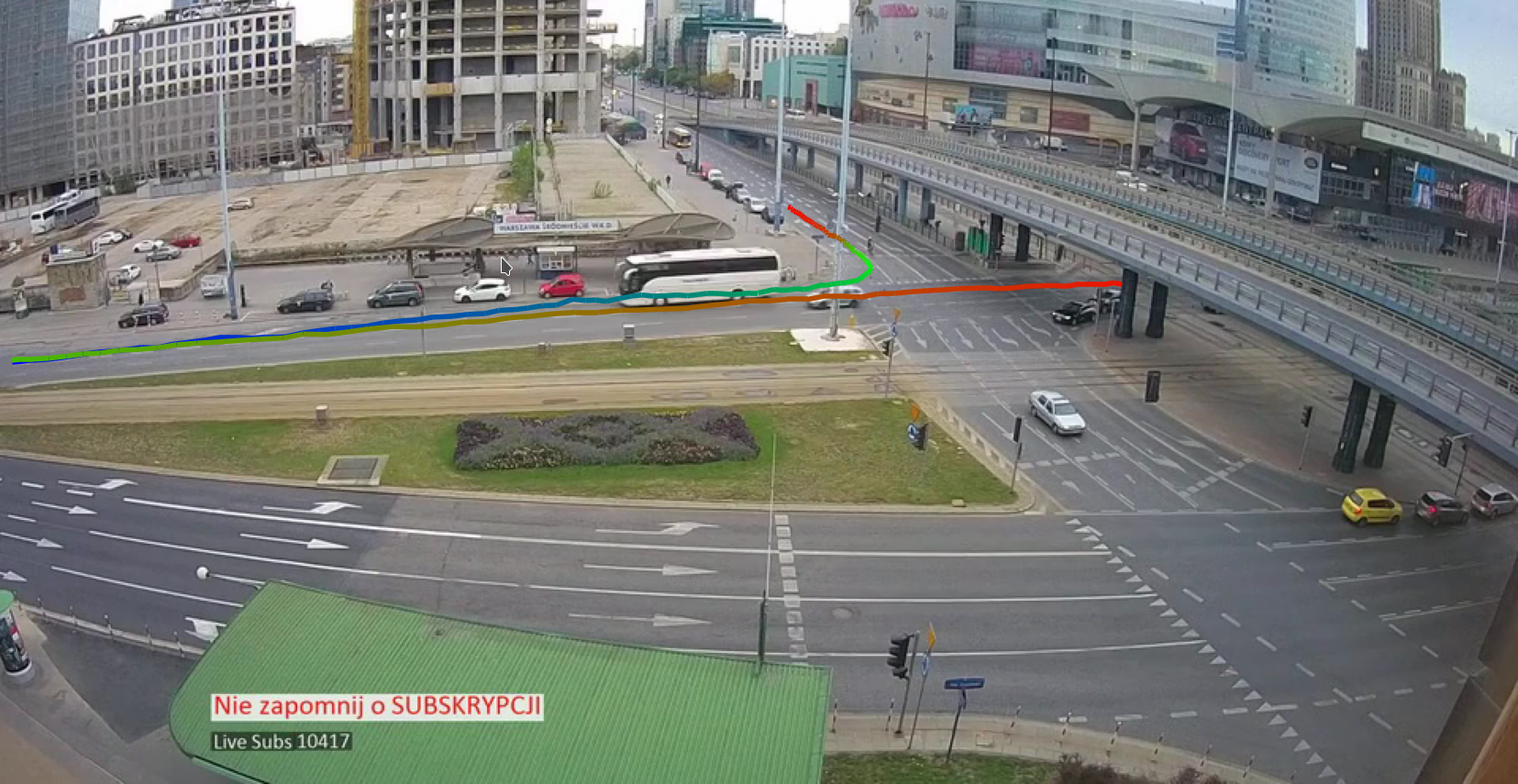} \vspace{10pt} \\
$\brackparen{\predbase{\textsf{InLane1}(A)} \seq \predbase{\textsf{Any}} \seq \predbase{\textsf{InLane2}(A)}} \wedge \predbase{\textsf{InLane2}(B)}$
\caption{A single match (top) for the multi-object query (bottom)
which captures one car, $A$, turning into a lane behind another car, $B$, that is in that lane. The trajectories change color from red to green as a function of time.
As can be seen, the car making the right turn does so just after the car going straight passes through the intersection.}
\label{fig:examplemulti}
\end{figure}

\para{Multi-object queries.}
So far, we have focused on queries that identify trajectories by processing each trajectory in isolation. A key feature of our framework is that users can express queries over multiple trajectories---for example,
\begin{align*}
\big(\basepred{\textsf{InLane1}(B)} \land \basepred{\textsf{ChangeLane2To1}(A)}\big) \seq \basepred{\textsf{InFront}(A,B)}.
\end{align*}
This query says that car $B$ is in lane 1 while car $A$ changes from lane 2 to lane 1, and car $A$ ends up in front of car $B$.  Note that the predicates now include variables indicating which object they refer to, and the predicate $\textsf{InFront}(A,B)$ refers to multiple objects. An example of a pair of trajectories selected by a multi-object query is shown in Figure~\ref{fig:examplemulti}.

\section{Query Language}
\label{sec:lang}

We describe our query language for matching object trajectories in videos. Our system first preprocesses the video using an object tracker to obtain trajectories, which are sequences $z=(x_0,x_1,...,x_{n-1})$ of states $x_i\in\mathcal{X}$. Then, a query $Q$ in our language maps each trajectory $z$ to a value $\mathbb{B}=\{0,1\}$ indicating whether it matches $z$. Our language is similar to both STL
and KAT.
One key difference is that predicates are over arbitrary subsequences of $z$ rather than single states $x$. In the main paper, we consider a simpler language, but in \refappendix{sec:langexts} we show how it can be extended to subsume both STL and KAT.

\para{Trajectories.}
We begin by describing the input to a query in our language, which is the representation of one or more concurrent object trajectories in a video.

Consider a space $\mathcal{S}$ corresponding to a single object detection in a single video frame---e.g., $s\in\mathcal{S}\subseteq\mathbb{R}^6$ might encode the 2D position, velocity, and acceleration of $s$ in image coordinates. When considering $m$ concurrent objects, let the space of \emph{states} $\mathcal{X} = \mathcal{S}^m$, and then a \emph{trajectory} $z\in\mathcal{Z} = \mathcal{X}^*$ is a sequence $z=(x_0,x_1,...,x_{n-1})$ of states of length $|z|=n$. We use the notation $z_{i:j}=(z_i,z_{i+1},...,z_{j-1})$ to denote a subtrajectory of $z$.

\para{Predicates.}
We assume a set of predicates $\Phi$ is given, where each predicate $\varphi\in\Phi$ matches trajectories $z\in\mathcal{Z}$; we use $\textsf{sat}_{\varphi}(z)\in\mathbb{B}=\{0,1\}$ to indicate that $\varphi$ matches $z$. As discussed below, queries in our language compose these predicates to match more complex patterns. 

Next, predicates in our language may have real-valued parameters that must be specified. We denote such a predicate $\varphi$ with parameter $\theta\in\mathbb{R}$ by $\varphi_\theta$. To enable our synthesis algorithm to efficiently synthesize these real-valued parameters, we leverage
the monotonicity in all such predicates we have used in our queries. In particular, we assume that the semantics of these predicates have the form
\begin{align*}
\llbracket\varphi_{\theta}\rrbracket(z)\coloneqq\mathbbm{1}(\iota_{\varphi}(z)\ge\theta),
\end{align*}
where $\iota_\varphi:\mathcal{Z}\to\mathbb{R}$ is a scoring function. We also assume that the range of $\iota_\varphi$ is bounded (which can be achieved with a sigmoid function, if necessary). For example, for the predicate $\textsf{DispLt}_{\theta}$, we have
$\iota_{\textsf{DispLt}}(z)=-\lVert z_0 - z_{n-1} \rVert$.
Thus, $\iota_{\textsf{DispLt}}(z)\ge\theta$ says the total displacement is at most $-\theta$.
We describe the predicates we include in Section~\ref{sec:setup}; they can easily be extended.

\para{Syntax.}
The syntax of our language is
\begin{align*}
Q &::= \varphi \mid Q \seq Q \mid Q^k \mid Q \land Q,
\end{align*}
where $Q^k=Q;Q;...;Q$ ($k$ times). That is, the base case is a single predicate $\varphi$, and queries can be composed using sequencing ($Q\seq Q$)
and conjunction ($Q\wedge Q$). Operators for disjunction, negation, Kleene star, and STL's ``until'' are discussed in \refappendix{sec:addops}. We describe constraints imposed on our language during synthesis in Section~\ref{sec:impl}.

\para{Semantics.}
The satisfaction semantics of queries have type
$\den{\cdot} : \mathcal{Q} \to \mathcal{Z} \to \mathbb{B}$,
where $\mathcal{Q}$ is the set of all queries in our language, $\mathcal{Z}$ is the set of trajectories, and $\mathbb{B}=\{0,1\}$.
In particular, $\den{Q}(z)\in\mathbb{B}$ indicates whether the query $Q$ matches trajectory $z$. The semantics are defined in Figure~\ref{fig:syntaxsemnatics}.
The base case of a single predicate $\varphi$ checks whether $\varphi$ matches $z$; conjunction $Q_1 \land Q_2$ checks if both conjuncts match; and sequencing $Q_1\seq Q_2$ checks if $z$ can be split into $z=z_{0:k}z_{k:n}$ in a way that $Q_1$ matches $z_{0:k}$ and $Q_2$ matches $z_{k:n}$. The semantics can be evaluated in time $O(\abs{Q} \cdot n^2)$.

\begin{figure}[t]
    \centering
    \begin{align*}
    \den{\varphi}(z) &\coloneqq \textsf{sat}_{\varphi}(z) \\
    \den{Q_1 \land Q_2}(z) &\coloneqq \den{Q_1}(z) \land \den{Q_2}(z) \\
    \den{Q_1 \seq Q_2}(z) &\coloneqq \bigvee_{k=0}^n ~ \den{Q_1}(z_{0:k}) \land \den{Q_2}(z_{k:n})
    \end{align*}
    \caption{Satisfaction semantics of our query language; $z\in\mathcal{Z}$ is a trajectory of length $n$ and $\varphi\in\Phi$ are predicates.
    Iteration ($Q^k$) can be expressed as repeated sequencing.
    }
    \label{fig:syntaxsemnatics}
\end{figure}

\section{Synthesis Algorithm}
\label{sec:algo}

We describe our algorithm for synthesizing queries consistent with a given set of examples. It performs a syntax-guided enumerative search over the space of possible queries~\cite{alur2013syntax}. In more detail, it enumerates \emph{sketches}, which are partial programs where only parameter values are missing. For each sketch, it uses a quantitative pruning strategy to compute the subset of the input parameters for which the resulting query is consistent with the given examples. A key contribution is how our algorithm uses quantitative semantics for quantitative pruning.

\subsection{Problem Formulation}
\label{sec:formulation}

\para{Partial queries.}
A \emph{partial query} is
in the grammar
\begin{align*}
Q ::=\; ?? \mid \varphi_{??} \mid \varphi \mid Q \seq Q \mid Q^k \mid Q \land Q.
\end{align*}
Note that there are two kinds of holes: (i) a \emph{predicate hole} $h=\;??$ that can be filled by a sub-query $Q$, and (ii) a \emph{parameter hole} $h=\varphi_{??}$ that can be filled by a real value $\theta_h\in\mathbb{R}$.
We denote the predicate holes of $Q$ by $\mathcal{H}_{\varphi}(Q)$, the parameter holes by $\mathcal{H}_{\theta}(Q)$, and let $\mathcal{H}(Q)=\mathcal{H}_{\varphi}(Q)\cup\mathcal{H}_{\theta}(Q)$.
A partial query $Q$ is a \emph{sketch} (denoted $Q\in\mathcal{Q}_{\text{sketch}}$)~\cite{solar2008program} if $\mathcal{H}_{\varphi}(Q)=\varnothing$, and is \emph{complete} (denoted $Q\in\bar{\mathcal{Q}}$) if $\mathcal{H}(Q)=\varnothing$.
For example, for $Q=\basepred{\textsf{DispLt}_{??1}}\wedge??2$, we have $\mathcal{H}_{\theta}(Q)=\{??1\}$ and $\mathcal{H}_{\varphi}(Q)=\{??2\}$. (We label each hole $h=\;??i$ with an identifier $i\in\mathbb{N}$ to distinguish them.)

\para{Refinements and completions.}
Given query $Q\in\mathcal{Q}$, predicate hole $h\in\mathcal{H}_{\varphi}(Q)$, and production $R=Q\to f(Q_1,...,Q_k)$
we can \emph{fill} $h$ with $R$ (denoted $Q'=\textsf{fill}(Q,h,R)$) by replacing $h$ with $f(??1,...,??k)$, where each $??i$ is a fresh hole, and similarly given a parameter hole $h\in\mathcal{H}_{\theta}(Q)$ and a value $\theta_h\in\mathbb{R}$.
We call $Q'$ a \emph{child} of $Q$ (denoted $Q\to Q'$).
Next, we call $Q''$ a \emph{refinement} of $Q$ (denoted $Q\xrightarrow{*}Q''$) if there exists a sequence $Q\to...\to Q''$; if furthermore $Q''\in\bar{\mathcal{Q}}$, we say it is a \emph{completion} of $Q$. For example, we have
\begin{align*}
??1
&\to \; ??2 \seq ??3
\to\predbase{\textsf{InLane1}} \seq ??3
\to \ldots.
\end{align*}
Here, $\predbase{\textsf{InLane1}} \seq ??3$ is a child (and refinement) of $??2 \seq ??3$ obtained by filling $??2$ with $Q\to\predbase{\textsf{InLane1}}$---i.e.,
\begin{align*}
\predbase{\textsf{InLane1}} \seq ??3
=\textsf{fill}(??2 \seq ??3, ??2, Q\to\predbase{\textsf{InLane1}}).
\end{align*}

\para{Parameters.}
We let $\theta\in\mathbb{R}^{|\mathcal{H}_\theta(Q)|}$ denote a choice of parameters for each $h\in\mathcal{H}_\theta(Q)$, let $\theta_h\in\Theta_h\subseteq\mathbb{R}$ denote the parameter for hole $h$, and let $Q_{\theta}$ denote the query obtained by filling each $h\in\mathcal{H}_{\theta}(Q)$ with $\theta_h$. Note that if $Q\in\mathcal{Q}_{\text{sketch}}$, then $Q_{\theta}\in\bar{\mathcal{Q}}$ is complete. For example, consider the sketch
\begin{align*}
Q=\basepred{\textsf{DispLt}_{??1}}\wedge\basepred{\textsf{MinLength}_{??2}}.
\end{align*}
This query has two holes, so its parameters are
$\theta\in\mathbb{R}^2$. If $\theta = (3.2, 5.0)$, then $\theta_{??1}=3.2$ is used to fill hole $??1$ and $\theta_{??2}=5.0$ is used to fill $??2$. In particular,
\begin{align*}
Q_{\theta}=\basepred{\textsf{DispLt}_{3.2}}\wedge\basepred{\textsf{MinLength}_{5.0}}.
\end{align*}

\para{Query synthesis problem.}
Given examples $W\subseteq\mathcal{W}=\mathcal{Z}\times\mathbb{B}$, where $\mathbb{B}=\{0,1\}$, our goal is to find a query $Q\in\bar{\mathcal{Q}}$ that correctly labels these examples---i.e.,
\begin{align*}
\psi_W(Q)\coloneqq\bigwedge_{(z,y)\in W}(\llbracket Q\rrbracket(z)=y).
\end{align*}
Thus, $\psi_W(Q)$ indicates whether $Q$ is consistent with the labeled examples $W$. Our goal is to devise a synthesis algorithm that is sound and complete---i.e., it finds a query that satisfies $\psi_W(Q) = 1$ if and only if one exists.

\subsection{Algorithm Overview}
Our algorithm enumerates sketches $Q\in\mathcal{Q}_{\text{sketch}}$; for each one, it tries to compute parameter values $\theta$ such that the completed query $Q_{\theta}$ is consistent with $W$---i.e., $\psi_W(Q_\theta)=1$. It can either stop once it has found a consistent query, or identify additional queries that are consistent with $W$. Algorithm~\ref{alg:main} shows this high-level strategy---at each iteration, it selects a sketch $Q$, determines a region $B$ of the parameter space containing consistent parameters $\theta\in B$, and adds $(Q,B)$ to a list of consistent queries that solve the synthesis problem.

The key challenge is searching over the space of continuous parameters $\theta$ for a given sketch $Q$ such that $Q_\theta$ is consistent with $W$. For efficiency, we rely heavily on pruning the search space. At a high level, consider evaluating a single candidate parameter $\theta$ on a single example $(z,y)\in W$---i.e., check whether $\llbracket Q_\theta\rrbracket(z)=y$. If this condition does not hold, then we can not only prune $\theta$ from the search space, but also a significant fraction of additional candidates. For instance, suppose $\llbracket Q_\theta\rrbracket(z)=1$ but $y=0$; if $\theta'\le\theta$ (in all components), then by a monotonicity property we prove for our semantics, we also have $\llbracket Q_{\theta'}\rrbracket(z)=1$. Thus, we can also prune $\theta'$.

Previous work has leveraged this property to prune the search space~\cite{learningmonotone,monotoniclogical,efficientenumerative}. Using a strategy based on binary search, for a given example $(z,y)\in W$, we can identify ``boundary'' parameters $\theta$ to accuracy $\varepsilon$ in $O(\log(1/\varepsilon))$ steps---i.e., compute $\theta$ for which $\llbracket Q_{\theta-\vec{\varepsilon}}\rrbracket(z)=1$ and $\llbracket Q_{\theta+\vec{\varepsilon}}\rrbracket(z)=0$.

Our algorithm avoids this binary search process, which can lead to
a significant speedup in practice. The key idea is to devise a quantitative semantics for queries that directly computes $\theta$; in fact, this quantitative semantics is closely related to robust temporal logic semantics, where the conjunction and disjunction of the satisfaction semantics are replaced with minimum and maximum, respectively.

\begin{algorithm}[t]
    \begin{algorithmic}[1]
    \Procedure{SynthesizeQuery}{$W$}
    \State $\mathcal{Q}_{\text{con}}\gets\varnothing$
    \For{$Q\in\mathcal{Q}_{\text{sketch}}$}
    \State $B\gets\textsf{SynthesizeParameters}(W,Q)$
    \State $\mathcal{Q}_{\text{con}}\gets\{(Q,B)\}$
    \EndFor
    \State \Return $\mathcal{Q}_{\text{con}}$
    \EndProcedure
    \end{algorithmic}
    \caption{Synthesizes consistent queries using the subroutine in Algorithm~\ref{alg:synthparams}}
    \label{alg:main}
\end{algorithm}

\subsection{Pruning with Boundary Parameters}

We begin by describing how ``boundary parameters'' can be used to prune a portion of the search space over parameters. First, for \emph{any} candidate parameters $\theta$, we can prune parameters $\theta'\le\theta$ (if $\den{Q_\theta}(z)=1$ and $y=0$) or $\theta'\ge\theta$ (if $\den{Q_\theta}(z)=0$ and $y=1$). Pruned regions of the parameter space take the form of hyper-rectangles, which we call \emph{boxes}. For convenience, let $\vec{\infty} \coloneqq (\infty, \ldots, \infty)$.
\begin{definition}
\rm
Given $x, y \in \bar{\mathbb{R}}^d$,
where $\bar{\mathbb{R}}=\mathbb{R}\cup\{\pm\infty\}$, a \emph{box} is an axis-aligned half-open hyper-rectangle
$\lfloor x, y \rceil \coloneqq \{v \mid x_i < v_i \leq y_i\}\subseteq\mathbb{R}^d$.
\end{definition}
The key property ensuring that parameters prune boxes of the search space is that the semantics are monotonically decreasing in $\theta$.
\begin{lemma}
\label{lem:montonic}
Given sketch $Q$, trajectory $z$, and two candidate parameters $\theta,\theta'\in\mathbb{R}^d$ such that $\theta\le\theta'$ component-wise, we have $\den{Q_\theta}(z)\ge\den{Q_{\theta'}}(z)$.
\end{lemma}
The proof follows by structural induction on the query semantics: the base case follows since the semantics $\mathbbm{1}(\iota_\varphi(z)\ge\theta_k)$ for predicates is monotonically decreasing in $\theta_k$, and the inductive case follows since conjunction and disjunction are monotonically increasing in their inputs (so they are also monotonically decreasing in $\theta_k$). Below, we show how monotonicity ensures that we can prune whole regions of the search space if we find boundary parameters.

As an example, suppose we have two trajectories, $z_0$ of a car driving quickly and then slowly, and $z_1$ of a car driving slowly and then quickly, and that we are trying to synthesize a query for $W=\{(z_0,0),(z_1,1)\}$. For simplicity, we assume both $z_0=(0.9,0.6)$ and $z_1=(0.5,0.8)$ have just two time steps each, with just a single component representing velocity. Furthermore, we assume there is just a single predicate $\expred{\theta}$ matching time steps where the velocity is at least $\theta$, where $\theta$ is a real-valued parameter. Since $\expred{\theta}$ matches single time steps, the satisfaction semantics is $0$ except on trajectories of length $1$, so:
\begin{align*}
\iota_{\expredname}((z_0)_{0:1}) &= 0.9  &  \iota_{\expredname}((z_0)_{1:2}) &= 0.6  &  \iota_{\expredname}((z)_{i:i}) &= -\infty \\
\iota_{\expredname}((z_1)_{0:1}) &= 0.5  &  \iota_{\expredname}((z_1)_{1:2}) &= 0.8  &  \iota_{\expredname}((z)_{0:2}) &= -\infty
\end{align*}
Consider the sketch $Q = \expred{??1}; \expred{??2}$. We can see that the candidate parameters $(0.5, 0.6)$ satisfy $\den{Q_{(0.5,0.6)}}(z_1)=1$:
\begin{align*}
\llbracket Q_{(0.5, 0.6)} \rrbracket((z_1)_{0:n}) &= \bigvee_{k=0}^2 \llbracket \expred{0.5} \rrbracket((z_1)_{0:k}) \land \llbracket \expred{0.6} \rrbracket((z_1)_{k:n})\\
&= \llbracket \expred{0.5} \rrbracket((z_1)_{0:1}) \land \llbracket \expred{0.6} \rrbracket((z_1)_{1:2})\\
&= \mathbbm{1}(0.5 \geq 0.5) \land \mathbbm{1}(0.8 \geq 0.6) \\
&= 1,
\end{align*}
where the second equality holds because $\expred{\theta}$ matches only length-$1$ trajectories, so the $k=0$ and $k=2$ cases evaluate to $0$. Since the semantics are monotonically decreasing, we have $\den{Q_\theta}(z_1)=1$ for any $\theta \in \boxlit{(-\infty, -\infty)}{(0.5, 0.6)}$.

Notice, however, that if we were to move any $\vec{\varepsilon} > 0$ upward, we would have $\den{Q_{(0.5 + \varepsilon_1, 0.6 + \varepsilon_2)}}(z_1) = \mathbbm{1}(0.5 \geq 0.5 + \varepsilon_1) \land \mathbbm{1}(0.8 \geq 0.6 + \varepsilon_2) = 0$. So we know $\den{Q_\theta}(z_1)=0$ for any $\theta \in \boxlit{(0.5, 0.6)}{(\infty, \infty)}$. This is because $(0.5, 0.6)$ lies on the boundary between $\{\theta'\mid\den{Q_{\theta'}}(z)=0\}$ and $\{\theta'\mid\den{Q_{\theta'}}(z)=1\}$. This boundary plays a key role in our algorithm.

\begin{definition}
\rm
Given a sketch $Q$ with $d$ parameter holes and a trajectory $z$, we say $\theta \in \mathbb{R}^d \cup \{\bot, \top\}$ is a \emph{boundary parameter} if one of the following holds:
\begin{itemize}
\item $\theta \in \mathbb{R}^d$ and $\llbracket Q_\theta \rrbracket(z) = 1$, but $\llbracket Q_{\theta'} \rrbracket(z) = 0$ for all $\theta' \in \boxlit{\theta}{\vec{\infty}}$
\item $\theta = \bot$ and $\llbracket Q_{\theta'} \rrbracket(z) = 0$ for all $\theta' \in \boxlit{-\vec{\infty}}{\vec{\infty}}$
\item $\theta = \top$ and $\llbracket Q_{\theta'} \rrbracket(z) = 1$ for all $\theta' \in \boxlit{-\vec{\infty}}{\vec{\infty}}$
\end{itemize}
\end{definition}
In the first case, by monotonicity, we also have $\den{Q_{\theta'}}(z)=1$ for all $\theta'\in\boxlit{-\vec{\infty}}{\theta}$; thus, $\theta$ lies on the boundary between parameters $\theta'$ where $Q_{\theta'}$ evaluates to $1$ and those where it evaluates to $0$. The second and third cases are where $Q_{\theta'}$ always evaluates to $0$ and $1$, respectively.

Given a boundary parameter $\theta$ for an example $(z,y)\in W$, we can prune $\boxlit{\theta}{\vec{\infty}}$ if $y=1$ or $\boxlit{-\vec{\infty}}{\theta}$ if $y=0$. Intuitively, boundary parameters provide optimal pruning along a fixed direction in the parameter space. Thus, our algorithm focuses on computing boundary parameters for pruning.

In Figure~\ref{fig:prune} (a), if $\theta_1$ is a boundary parameter for $z_1$, we know that the blue region satisfies $z_1$, and thus is consistent with the label $1$, while the red region dissatisfies $z_1$, and thus is inconsistent with the label $1$. Similarly, in Figure~\ref{fig:prune} (b), if $\theta_0$ is a boundary parameter for $z_0$, we know that the red satisfies $z_1$, and thus is inconsistent with the label $0$, while the blue dissatisfies $z_0$, and thus is consistent with the label $0$.

\begin{figure*}[t]
\centering
\includegraphics[width=\textwidth]{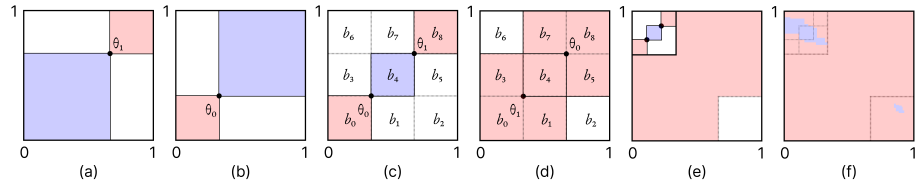}
\caption{(a) shows a boundary parameter, $\theta_1$, for $z_1$, and a region that is inconsistent with $z_1$ and can be pruned (red), as well as a region that is consistent with it (blue). (b) similarly shows a boundary parameter $\theta_0$ for $z_0$. (c) shows the pruning pair composed of $\theta_0$ and $\theta_1$, a region consistent with both (blue), and regions inconsistent with either (red). (d) is the same as (c), but if $\theta_0$ and $\theta_1$ swapped places. The labels $b_0$ through $b_8$ denote analogous boxes in (c) \& (d). (e) shows how, if (d) were the result of the first step of search and $b_6$ were chosen next, search could proceed. (f) shows ground truth consistent (blue) and inconsistent (red) regions that the search process in (d) \& (e) might converge toward.}
\label{fig:prune}
\end{figure*}

\subsection{Pruning with Pairs of Boundary Parameters}

To extend pruning to the entire dataset $W$, we could simply prune the union of the individual pruned regions for each $(z,y)\in W$. However, one important feature of our approach is that we can also establish regions of the parameter space where the parameters are guaranteed to be consistent with $W$. To formalize this idea, we introduce the concept of a ``pruning pair'', which is a pair of boundary parameters which might allow us to find such a consistent region.
\begin{definition}
\rm
Given a sketch $Q$ and a dataset $W$, a pair of boundary parameters $\theta^-, \theta^+ \in \mathbb{R}^d \cup \{\bot, \top\}$ is a \emph{pruning pair} for $Q$ and $W$ if all of the following hold:
\begin{itemize}
\item $\theta^+$ is a boundary parameter for some $z \in W^+$ and, for all $z'\in W^+$ such that $z' \neq z$, we have $\llbracket Q_{\theta^+} \rrbracket(z') = 1$.
\item $\theta^-$ is a boundary parameter for some $z \in W^-$ and, for all $z'\in W^-$ such that $z' \neq z$, we have $\llbracket Q_{\theta^-} \rrbracket(z') = 0$.
\item $\theta^- < \theta^+$ or $\theta^- \geq \theta^+$.
\end{itemize}
If $\theta^- < \theta^+$, the pruning pair $(\theta^-, \theta^+)$ is \emph{consistent}, and \emph{inconsistent} otherwise.
\end{definition}
Our algorithm searches for pruning pairs along a fixed direction---i.e., it considers a curve $L\subseteq\mathbb{R}^d$ and looks for the following pruning pair along $L$:
\begin{align*}
\theta^+=\sup\left\{\theta\in L\biggm\vert\bigwedge_{z\in W^+}\den{Q_\theta}(z)=1\right\},
\quad
\theta^-=\inf\left\{\theta\in L\biggm\vert\bigwedge_{z\in W^-}\den{Q_\theta}(z)=0\right\}.
\end{align*}
Intuitively, $\theta^+$ is the largest $\theta$ that correctly classifies all positive examples, and conversely for $\theta^-$. We restrict to curves $L$ that are monotonically increasing in all components, in which case the supremum and infimum are well defined since $L$ comes with a total ordering (from its smallest point to its largest) that is consistent with the standard partial order on $\mathbb{R}^d$. Then, $(\theta^-,\theta^+)$ form a pruning pair: since $\theta^+$ is a boundary parameter for $z$, if we take $\theta^+$ to be any larger, then we must have $\den{Q_\theta}(z)=0$ for some $z\in W^+$, and similarly for $\theta^-$.

Given a curve $L$, we can compute an approximation to $\theta^+$ and $\theta^-$ via binary search. However, our algorithm avoids the need to do so by directly computing $\theta^+$ and $\theta^-$ using a quantitative semantics, which we describe in Section~\ref{sec:quantitative}.

Figure~\ref{fig:prune} (c) shows how the pair of boundary parameters $\theta_0$ for $z_0$ and $\theta_1$ for $z_1$ (where $L$ is the diagonal line) prunes the parameter space. The blue region is guaranteed to be consistent with $W$, as it is the intersection of the region below $\theta^+$, which must satisfy $\den{Q_\theta}(z_1)=1$, and the region above $\theta^-$, which must satisfy $\den{Q_\theta}(z_0)=0$. Conversely, the red regions are inconsistent with either $z_0$ or $z_1$, and therefore with $W$.
Thus, the red regions can be pruned, whereas the blue regions are solutions to our synthesis problem. Note that the red region is the union of the red regions in Figure~\ref{fig:prune} (a) \& (b), whereas the blue region is the intersection of the blue regions in Figure~\ref{fig:prune} (a) \& (b).

This pattern holds for any consistent pruning pair ($\theta^- < \theta^+$); if instead the pair is inconsistent ($\theta^- \geq \theta^+$), then the resulting pattern is illustrated in Figure~\ref{fig:prune} (d); in this case, we can prune the red regions as before, but there is no blue region of solutions.
In general, for a $d$ dimensional parameter space, a pruning pair divides the parameter space into $3^d$ boxes (i.e., for each dimension, the box can be below, in line with, or above the center box). The regions below $\theta^-$ and above $\theta^+$ can be pruned, and the region between $\theta^-$ and $\theta^+$ (if one exists) contains synthesis solutions. Precisely, it follows from the definitions and monotonicity that:

\begin{lemma}
\label{lem:prunesound}
Every $\theta \in \boxlit{-\vec{\infty}}{\theta^-}$ and $\theta \in \boxlit{\theta^+}{\infty}$ is inconsistent with $W$, and every $\theta \in \boxlit{\theta^-}{\theta^+}$ is consistent with $W$.
\end{lemma}
The remaining boxes need to be further analyzed by our algorithm.

\subsection{Pruning Parameter Search Algorithm}
\begin{algorithm}[t]
\begin{algorithmic}[1]
\Procedure{SynthesizeParameters}{$W, Q$}
\State $B_{\text{con}}\gets\varnothing$; $B_{\text{inc}}\gets\varnothing$, $B_{\text{unk}}\gets\{b_{\operatorname{initial}}\}$
\For{$i\in\{1,...,N\}$}
\State $b\gets\textsf{Pop}(B_{\text{unk}})$
\State $\theta^-, \theta^+ \gets \textsf{ComputePruningPair}(W,Q,b)$
\State $b_{\text{center}} \gets \boxlit{\min\{\theta^-, \theta^+\}}{\max\{\theta^-, \theta^+\}}$
\State $b_{\text{lower}}, b_{\text{upper}}, B_{\text{incomp}}, B_{\text{extra}} \gets \textsf{SplitBox}(b, b_{\text{center}})$
\If{$\theta^- < \theta^+$}
\State $B_{\text{con}} \gets B_{\text{con}} \cup \{b_{\text{center}}\}$
\State $B_{\text{inc}} \gets B_{\text{inc}} \cup \{b_{\text{lower}}, b_{\text{upper}}\}$
\State $B_{\text{unk}} \gets B_{\text{unk}} \cup B_{\text{incomp}} \cup B_{\text{extra}}$
\ElsIf{$\theta^- \geq \theta^+$}
\State $B_{\text{inc}} \gets B_{\text{inc}} \cup \{b_{\text{center}}, b_{\text{lower}}, b_{\text{upper}}\} \cup B_{\text{extra}}$
\State $B_{\text{unk}} \gets B_{\text{unk}} \cup B_{\text{incomp}}$
\EndIf
\EndFor
\State \Return $B_{\text{con}}$
\EndProcedure
\end{algorithmic}
\caption{Synthesizes consistent parameters for a given sketch}
\label{alg:synthparams}
\end{algorithm}

Next, we describe Algorithm~\ref{alg:synthparams}, which searches over the space of parameters to fill a sketch $Q$ for a given dataset $W$. The algorithm uses a subroutine that takes a box and returns a pruning pair in that box, which we describe in Section~\ref{sec:quantitative}. Given this subroutine, our algorithm maintains a work-list of ``unknown'' boxes (i.e., unknown whether parameters in these boxes are consistent or inconsistent with $W$). At each iteration, it pops a box from the work-list (in first-in-first-out order), uses the given subroutine to find a pruning pair inside that box, applies the pruning procedure described in the previous section, and then adds each new unknown box to the work-list.

For the last step, the current box $b$ is divided into $3^d$ smaller boxes.  The box $b_{\text{center}}\coloneqq\boxlit{\min\{\theta^-,\theta^+\}}{\max\{\theta^-,\theta^+\}}$ is pruned (added to $B_{\text{inc}}$) if the pair $(\theta^-,\theta^+)$ is inconsistent, and contains solutions to the synthesis problem otherwise (added to $B_{\text{con}}$). The boxes $b_{\text{lower}}=\boxlit{-\vec{\infty}}{\min\{\theta^-,\theta^+\}}$ and $b_{\text{upper}}=\boxlit{\max\{\theta^-,\theta^+\}}{\vec{\infty}}$ are always pruned. The boxes $b\in B_{\text{incomp}}$ are the remaining corners of $b$, and always have indeterminate consistency (added to $B_{\text{unk}})$. The remaining boxes $b\in B_{\text{extra}}$ are indeterminate if $(\theta^-,\theta^+)$ is consistent, and inconsistent otherwise. In our example, if the first step of the algorithm yielded Figure~\ref{fig:prune} (d), then the second step might pop $b_6$ and yield Figure~\ref{fig:prune} (e).

The following soundness result follows directly from Lemma~\ref{lem:prunesound}.%
\begin{theorem}
\label{thm:algsound}
In Algorithm~\ref{alg:synthparams}, every $\theta\in B_{\operatorname{con}}$ is consistent with $W$ for $Q$, and every $\theta\in B_{\operatorname{inc}}$ inconsistent.
\end{theorem}
In addition, the algorithm is complete for almost all parameters:
\begin{theorem}
\label{thm:algcomplete}
The Lebesgue measure of $\left\{\theta \in b \mid b \in B_{\operatorname{unk}}\right\} \to 0$ as $N\to\infty$.
\end{theorem}
See \refappendix{sec:proof:algcomplete} for the proof. In other words, all parameters outside a subset of measure zero are eventually classified as consistent or inconsistent; intuitively, the parameters that may never be classified are the ones along the decision boundary. This result holds since at any search depth, the fraction of the parameter space pruned can be lower-bounded.

\subsection{Computing Pruning Pairs via Quantitative Semantics}
\label{sec:quantitative}

The pruning algorithm depends on the ability to compute, given a box $b$, a pruning pair $(\theta^-, \theta^+)$ on the restriction of the parameter space to $b$. Recall that $\theta^+$ must be a boundary parameter for some $z^+ \in W^+$ and must satisfy $\den{Q_{\theta^+}}(z)=1$ for all other $z \in W^+$, and $\theta^-$ must be a boundary parameter for some $z^- \in W^-$, and must satisfy $\den{Q_{\theta^-}}(z)=0$ for all other $z \in W^-$.

Given a box $b=\boxlit{\theta_{\text{min}}}{\theta_{\text{max}}}$, our algorithm takes $L\subseteq\mathbb{R}^d$ to be the diagonal from $\theta_{\text{min}}$ to $\theta_{\text{max}}$ and computes the pruning pair along $L$. We can na\"{i}vely use binary search: for $\theta^+$, we search for the parameters where $\bigwedge_{z \in W^+} \den{Q_\theta}(z)$ transitions from $0$ to $1$, and similarly for $\theta^-$ and $\bigwedge_{z \in W^-} \neg \den{Q_\theta}(z)$.

Instead, by leveraging a quantitative semantics, we can directly compute $\theta^+$ and $\theta^-$, thereby reducing computation time substantially. Given a sketch $Q$, trajectory $z$, parameter $v \in \mathbb{R}^d$, and positive vector $u \in \mathbb{R}_{>0}^d$, we devise a quantitative semantics $\denq{Q}_{v,u}(z) \in \bar{\mathbb{R}}$ such that the parameter $\den{Q}_{v,u}(z) \cdot u + v$ is a boundary parameter. Intuitively, this semantics computes, starting at $v$, how many $u$-sized steps must be taken to reach the boundary. (For the uses in our algorithm, the number of steps is always in $[0,1]$.)
Then, for a box $b = \boxlit{\theta_{\text{min}}}{\theta_{\text{max}}}$, we can take $v = \theta_{\text{min}}$ and $u = \theta_{\text{max}}-\theta_{\text{min}}$, and compute
\begin{align*}
\theta^+ = \left(\min_{z \in W^+} \denq{Q}_{v,u}(z)\right) \cdot u + v,
\qquad
\theta^- = \left(\max_{z \in W^-} \denq{Q}_{v,u}(z)\right) \cdot u + v.
\end{align*}
We define the quantitative semantics in Figure~\ref{fig:quantsemantics}. The base case of $\varphi_{??}$ adjusts and rescales $\iota_\varphi$ by $v$ and $u$, and the other cases replace conjunction and disjunction in the satisfaction semantics with minimum and maximum. We have the following key result (where $\infty \cdot u \coloneqq \top$, $-\infty \cdot u \coloneqq \bot$, $\top + v \coloneqq \top$, and $\bot + v \coloneqq \bot$):
\begin{theorem}
\label{thm:quantcorrect}
For a sketch $Q$, trajectory $z$, parameter $v \in \mathbb{R}^d$, and positive vector $u \in \mathbb{R}_{>0}^d$, we have that $\denq{Q}_{v,u}(z) \cdot u + v$ is a boundary parameter of $z$ for $Q$.
\end{theorem}
See \refappendix{sec:proof:quantcorrect} for the full proof. For intuition, consider $\theta_{\text{min}}=\vec{0}$ and $\theta_{\text{max}}=\vec{1}$ (i.e., the current box $b\subseteq\mathbb{R}^d$ is the unit hypercube). Then, $v=\vec{0}$ and $u=\vec{1}$, so $\denq{Q}_{v,u}$ reduces to the standard max-min quantitative semantics
for temporal logic~\cite{fainekos2009robustness}.

\begin{figure}[t]
    \centering
    \begin{align*}
    \denq{\varphi_{??i}}_{v,u}(z) &\coloneqq \frac{\iota_{\varphi}(z) - v_i}{u_i} \\
    \denq{\varphi}_{v,u}(z) &\coloneqq \begin{cases} \infty &\text{if } \textsf{sat}_{\varphi}(z) = 1 \\ -\infty &\text{if }\textsf{sat}_{\varphi}(z) = 0.\end{cases} \\
    \denq{Q_1 \land Q_2}_{v,u}(z) &\coloneqq \min\{\denq{Q_1}_{v,u}(z),\denq{Q_2}_{v,u}(z)\} \\
    \denq{Q_1 \seq Q_2}_{v,u}(z) &\coloneqq \max_{0 \leq k \leq n}\min\{\denq{Q_1}_{v,u}(z_{0:k}),\denq{Q_2}_{v,u}(z_{k:n})\}
    \end{align*}
    \caption{The quantitative semantics of our language, taking a sketch $Q$, trajectory $z$, parameter $v \in \mathbb{R}^d$, and positive vector $u \in \mathbb{R}_{>0}^d$. $n$ is the length of $z$.}
    \label{fig:quantsemantics}
\end{figure}

Now, if we consider the satisfaction semantics of a base predicate $\den{\varphi_{\theta_i}}=\mathbbm{1}(\iota_{\varphi}(z)\ge\theta_i)$, then the value of $\theta_i$ where the sementics flips is just $\iota_\varphi(z)$. So any parameter with $i$-th component $\iota_\varphi(z)$ is a boundary parameter, and since $L$ has the same slope in all dimensions, the boundary parameter along $L$ is $\iota_\varphi(z) \cdot \vec{1} + \vec{0} = \denq{\varphi_{??i}}_{\vec{0},\vec{1}}(z) \cdot \vec{1} + \vec{0}$.

In the inductive cases, it suffices to show that we can replace conjunction and disjunction with minimum and maximum in the semantics. Since the satisfaction semantics is monotonically decreasing, as we move upward along $L$, at some point we will transition from $1$ to $0$. A conjunction becomes $0$ when either conjunct becomes $0$, so the transition will occur when we hit the first of the conjuncts' transition points (their minimum). Dually, a disjunction becomes $0$ when both disjuncts become $0$, so we will transition at the last of the disjuncts' transition points (their maximum).

Finally, the intuition behind $u$ and $v$ is that they ``preprocess'' the parameters so that we evaluate along the diagonal of the box $\boxlit{v}{v + u}$ instead of $\boxlit{\vec{0}}{\vec{1}}$.

\subsection{Implementation}
\label{sec:impl}
We implement our approach in a system called \toolname. It begins by running Algorithm~\ref{alg:main} on a small number of labeled examples.

\para{Active learning.}
With a small number of examples, there are typically many queries that are consistent with the labels, and yet which disagree on the labels of the remaining data. To disambiguate, we use an active learning strategy, asking the user to label specific trajectories that we choose, which are then added to our set of labeled examples. Queries that are not consistent with the new label are discarded. The labeling process continues until the set of consistent queries agrees on the labels of all unlabeled data.

When choosing the trajectory $z^*$ to query the user for next, we select the one on which the set of consistent queries $C$ disagrees most---i.e.,
$$z^* = \operatorname*{\arg\min}_{z\in Z}\left|J(z) - \frac{1}{2}\right|,$$
where
$$J(z)\coloneqq|C|^{-1}\sum_{Q_{\theta}\in C}\mathbbm{1}\left(\psi_{(z,y)}(Q_{\theta})\right)$$ is the fraction of consistent queries that predict a positive label for trajectory $z$. 

\para{Search implementation.}
In some cases, searching for consistent parameters may take a very long time. To improve performance, we impose a timeout: for each sketch, we pause search if either: (i) we find some consistent box of parameters or (ii) we've exceeded 25 steps. In both cases, we save the sets of consistent, inconsistent, and unknown boxes. At each step of active learning, the newly labeled example may render previously consistent parameters inconsistent, so we mark all consistent boxes as unknown. We then resume search, again until (i) we find some consistent box (which may be the same one we had before), or (ii) we again exceed 25 steps.

Note that while this timeout may cause us to query the user more often than is strictly necessary, it does not affect either the soundness or completeness of our approach, as we continue search after querying the user.

\para{Complete query selection.}
Active learning and evaluation of $F_1$ scores (in Section~\ref{sec:exp}) both require complete queries with specific parameters, rather than sketches with boxes of parameters. Since the set $C$ of consistent queries is infinitely large, we instead we use one query for each sketch that is known to have consistent parameters (sketches where search timed-out are thus not included). For those sketches, we pick the middle of the box of known-consistent parameters.

\section{Evaluation}
\label{sec:exp}
We demonstrate how our approach can be used to synthesize queries to solve interesting tasks: in particular, we show that (i) given just a few initial examples, it can synthesize queries that achieve good performance on a held-out test set, and (ii) our optimizations significantly reduce the synthesis time.

\subsection{Experimental Setup} \label{sec:setup}
\para{Datasets.} We evaluate on two datasets of object trajectories: YTStreams~\cite{bastani2020miris}, consisting of video and extracted object trajectories from fixed-position traffic cameras, and MABe22~\cite{mabe22}, consisting of trajectories of up to three mice interacting in a laboratory setting. We also evaluate on a synthetic maritime surveillance task from the STL synthesis literature~\cite{kong2017}. 
On YTStreams, we use two traffic cameras, one in Tokyo and one in Warsaw, and we consider single cars or pairs of cars. On MABe22, we consider pairs of mice. For the predicates used, see Table~\ref{tab:predicates:yts} for YTStreams, \refappendixtab{tab:predicates:mabe} for MABe22, and \refappendixtab{tab:predicates:ms} for maritime surveillance.

\begin{table*}[t]
    \centering
    \scriptsize
    \caption{The predicates used for the YTStreams dataset.}
    \begin{tabular}{>{\centering\arraybackslash}p{0.2\textwidth}p{0.75\textwidth}}
    \toprule
    {\bf Predicate} & \multicolumn{1}{c}{{\bf Description}} \\
    \midrule
    $\textsf{InLaneK}(A)$ & Whether, at every time-step in the interval, object A is sufficiently close to the annotated curve for lane K and A’s movement direction is sufficiently in line with the curve for K. \vspace{5pt} \\
    $\textsf{DurationNotShort}$ & Whether the interval spans at least 5 seconds. \vspace{5pt} \\
    $\textsf{AvgAccelGt}_{\theta}$ & Whether the average acceleration over the interval is at least $\theta$. \vspace{5pt} \\
    $\textsf{DistanceLt}_{\theta}$ & Whether, at every time-step in the interval, the distance between the two objects is less than $\theta$. \vspace{5pt} \\
    $\textsf{SpeedRatioGt}_{\theta}(A,B)$ & Whether, at every time-step in the interval, the speed of A divided by the speed of B is at least $\theta$. \vspace{5pt} \\
    $\textsf{DispLt}_{\theta}(A)$ & Whether the distance between the position in the first frame of the interval and the position in the last frame is less than $\theta$. \\
    \bottomrule
    \end{tabular}
    \label{tab:predicates:yts}
\end{table*}

\para{Tasks.}
On YTStreams, we manually wrote $5$ ground truth queries. Several queries apply to multiple configurations (e.g., different pairs of lanes), resulting in $10$ queries total (tasks H-Q in Table~\ref{tab:queries:yts}). The real-valued parameters were chosen manually, by visually examining whether they were selecting the desired trajectories. These queries cover a wide range of behaviors; for instance, they can capture behaviors such as human drivers making unprotected turns, an important challenge for autonomous cars~\cite{sadigh2016planning}, as well as cars trying to pass~\cite{schmerling2018multimodal}. We show examples of trajectories selected by three of our multi-object queries in Figure~\ref{fig:trajectories}. MABe22 describes $9$ queries for scientifically interesting mouse behavior. We implemented the $6$ most complex to use as ground truth queries (tasks A-F in \refappendixtab{tab:queries:mabe}). The maritime surveillance task has trajectory labels and so does not need a ground truth query (task G).

\begin{table*}[t]
    \scriptsize
    \centering
    \caption{Ground-truth queries for the YTStreams dataset. ``IDs'' indicates which tasks are instances of a given query. Multiple instantiations correspond to different lanes being used for ``lane 1'' and ``lane 2''. The first is a one-object Shibuya query, the second is a one-object Warsaw query, and the rest are two-object Warsaw queries.}
    \begin{tabular}{>{\centering\arraybackslash}p{0.15\textwidth}>{}p{0.83\textwidth}}
    \toprule
    \textbf{IDs} & \textbf{Query} \\
    \midrule
    H, I, J, K &
    Matches cars that turn, starting in lane 1 and ending in lane 2.

    \vspace{0.5em}
    \centerline{$\predbase{\textsf{InLane1}(A)} \seq \predbase{\textsf{Any}} \seq \predbase{\textsf{InLane2}(A)}$}
    \\
    L, M &
    Matches cars that accelerate for a significant period of time while in lane 1.

    \vspace{0.5em}
    \centerline{$\predbase{\textsf{InLane1}(A)} \land \predbase{\textsf{AvgAccelGt}(A)}_{??} \land \predbase{\textsf{DurationNotShort}}$} \\
    N &
    Matches pairs of cars where car B is in lane 2 for the entire duration of A turning from lane 1 into lane 2.

    \vspace{0.5em}
    \centerline{$\left(\predbase{\textsf{InLane1}(A)} \seq \predbase{\textsf{Any}} \seq \predbase{\textsf{InLane2}(A)}\right) \land \predbase{\textsf{InLane2}(B)}$} \\
    O, P &
    Matches pairs of cars in parallel lanes, 1 and 2, where car A is going faster than car B for a significant period of time.

    \vspace{0.5em}
    \centerline{$\predbase{\textsf{InLane1}(A)} \land \predbase{\textsf{InLane2}(B)} \land \predbase{\textsf{DurationNotShort}} \land \predbase{\textsf{SpeedFactorGt}(A,B)}_{??}$} \\
    Q &
    Matches pairs of cars in parallel lanes, 1 and 2, where the cars are close for a significant period of time. 

    \vspace{0.5em}
    \centerline{$\predbase{\textsf{InLane1}(A)} \land \predbase{\textsf{InLane2}(B)} \land \predbase{\textsf{DurationNotShort}} \land \predbase{\textsf{DistanceLt}(A,B)}_{??}$} \\
    \bottomrule
    \end{tabular}
    \label{tab:queries:yts}
\end{table*}

\begin{figure*}[t]
    \centering
    \includegraphics[width=0.32\textwidth]{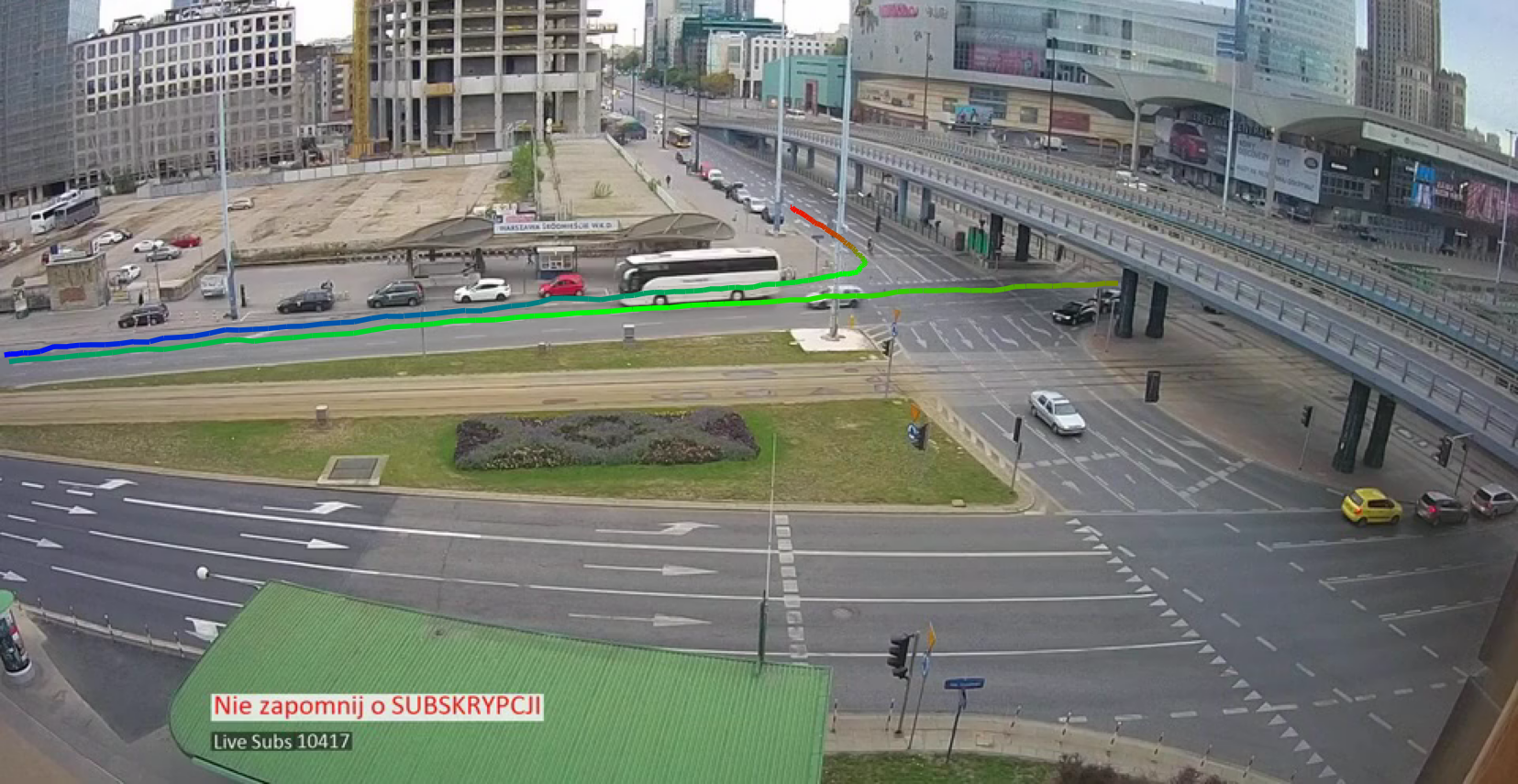}
    \includegraphics[width=0.32\textwidth]{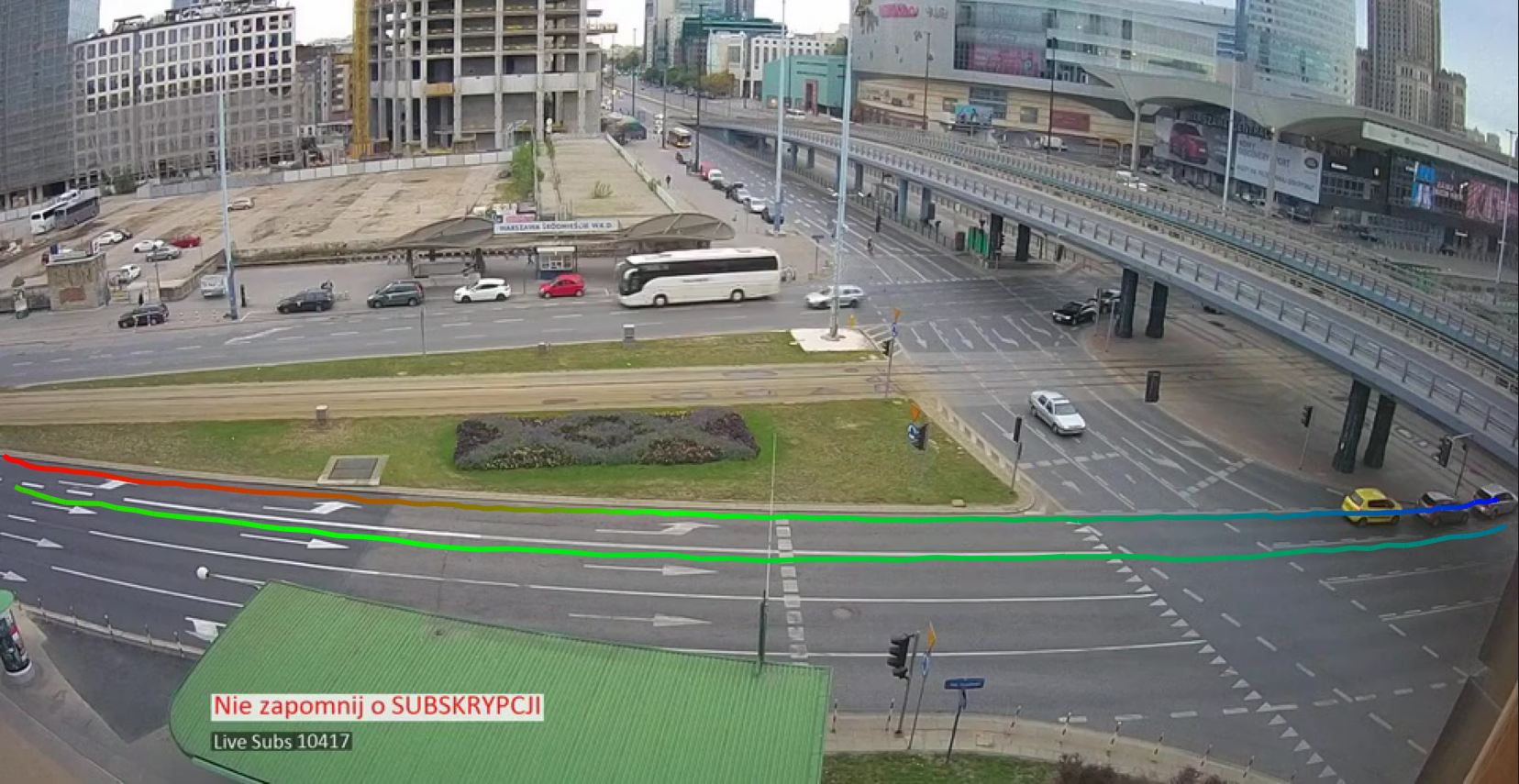}
    \includegraphics[width=0.32\textwidth]{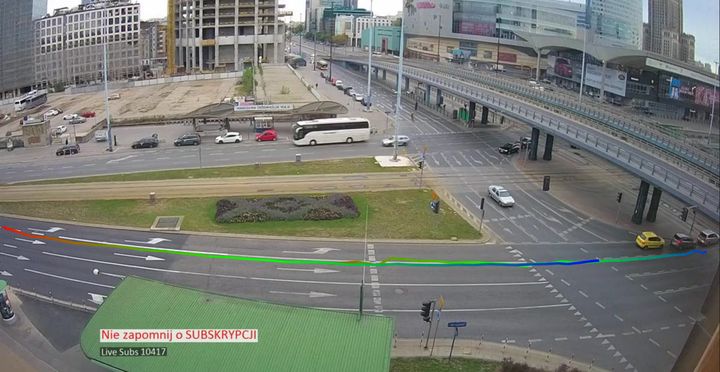}
    \caption{Trajectories selected by multi-object queries. Each image shows two objects; the color of each one changes from red to green to denote the progression of time. Left: Unprotected right turn into lane with oncoming traffic. Middle: Bottom car drives faster than the top one and passes it. Right: One car driving closely behind the other.}
    \label{fig:trajectories}
\end{figure*}

\para{Synthesis.}
For each task, we divide the set $Z$ of all trajectories into a train set $Z_\text{train}$ and a test set $Z_\text{test}$, using trajectories in the first half of the video for training, and those in the second half for testing. We randomly sample a set of initial labeled examples $W$ from $Z_{\text{train}}$, with 2 samples being positive and 10 being negative, and then actively label 25 additional examples from $Z_{\operatorname{train}}$. For YTStreams and MABe22, labels are from the ground truth query.
For tractability, we limit search to sketches with at most three predicates, at most two of which may have parameters. In most cases, this excludes the ground truth from the search space.

\subsection{Accuracy of Synthesized Queries} \label{sec:accuracy}

We show that \toolname synthesizes accurate queries from just a few labeled examples.
We evaluate the $F_1$ score of the synthesized queries on $Z_{\text{test}}$. Recall that our algorithm returns a list $C$ of consistent queries; we report the median $F_1$ score across $Q\in C$.

\para{Baselines.}
We compare to (i) an ablation where we replace our active learning strategy with an approach that labels $z$ uniformly at random from the remaining unlabeled training examples; (ii) an LSTM~\cite{hochreiter1997long,argoverse} neural network; and (iii) a transformer neural network~\cite{vaswani2017,giuliari2020transformer,franco2023transformer}. Because neural networks perform poorly on such small datasets, we pretrain the LSTM on an auxiliary task, namely, trajectory forecasting~\cite{kitani2012activity}.
Then, we freeze the hidden representation of the learned LSTM, and use these as features to train a logistic regression model on our labeled examples. The neural network baselines do active learning by selecting among the unlabeled trajectories the one with the highest predicted probability of being positive.

\para{Results.}
We show the $F_1$ score of each of the 17 queries in Table~\ref{tab:al_f1} after 0, 5, 10, and 25 steps of active learning. After just 10 steps, our approach provides $F_1$ score above 0.99 on 10 of 17 queries, and after 25 steps, it yields an $F_1$ score above 0.9 on all but 2 queries.
Thus, \toolname is able to synthesize accurate queries with relatively little user input.
The neural networks achieve poor performance, particularly on the more difficult queries.

\begin{table*}[t]
\centering
\notsotiny
\caption{$F_1$ score after $n$ steps of active learning, with our algorithm for selecting tracks to label (``$Q$''), an active learning ablation (``$R$''), an LSTM (``$L$''), and a transformer (``$T$''). For $Q$ and $R$, there may be many queries consistent with the labeled data, so the median $F_1$ score is reported. Bold indicates best score at a given number of steps.}
\begin{tabular}{crrrrrrrrrrrrrrrr}
\toprule \multicolumn{1}{c}{\multirow{2}{*}{\textbf{ID}}} & \multicolumn{4}{c}{\textbf{0 Steps}} & \multicolumn{4}{c}{\textbf{5 Steps}} & \multicolumn{4}{c}{\textbf{10 Steps}} & \multicolumn{4}{c}{\textbf{25 Steps}} \\
& \multicolumn{1}{c}{$Q$} & \multicolumn{1}{c}{$R$} & \multicolumn{1}{c}{$L$} & \multicolumn{1}{c}{$T$} & \multicolumn{1}{c}{$Q$} & \multicolumn{1}{c}{$R$} & \multicolumn{1}{c}{$L$} & \multicolumn{1}{c}{$T$} & \multicolumn{1}{c}{$Q$} & \multicolumn{1}{c}{$R$} & \multicolumn{1}{c}{$L$} & \multicolumn{1}{c}{$T$} & \multicolumn{1}{c}{$Q$} & \multicolumn{1}{c}{$R$} & \multicolumn{1}{c}{$L$} & \multicolumn{1}{c}{$T$}\\
\midrule
A & \hspace{0.01in} 0.69 & \hspace{0.01in} 0.69 & \hspace{0.01in} \bf{1.00} & \hspace{0.01in} 0.74 & \hspace{0.01in} 1.00 & \hspace{0.01in} \bf{1.00} & \hspace{0.01in} 1.00 & \hspace{0.01in} 0.74 & \hspace{0.01in} \bf{1.00} & \hspace{0.01in} 1.00 & \hspace{0.01in} 1.00 & \hspace{0.01in} 0.74 & \hspace{0.01in} \bf{1.00} & \hspace{0.01in} 1.00 & \hspace{0.01in} 1.00 & \hspace{0.01in} 0.74 \\
B & \hspace{0.01in} \bf{0.99} & \hspace{0.01in} \bf{0.99} & \hspace{0.01in} 0.47 & \hspace{0.01in} 0.20 & \hspace{0.01in} \bf{0.99} & \hspace{0.01in} \bf{0.99} & \hspace{0.01in} 0.47 & \hspace{0.01in} 0.25 & \hspace{0.01in} \bf{0.99} & \hspace{0.01in} \bf{0.99} & \hspace{0.01in} 0.47 & \hspace{0.01in} 0.05 & \hspace{0.01in} \bf{0.99} & \hspace{0.01in} 0.98 & \hspace{0.01in} 0.47 & \hspace{0.01in} 0.06 \\
C & \hspace{0.01in} \bf{0.96} & \hspace{0.01in} \bf{0.96} & \hspace{0.01in} 0.38 & \hspace{0.01in} 0.09 & \hspace{0.01in} \bf{0.99} & \hspace{0.01in} 0.96 & \hspace{0.01in} 0.38 & \hspace{0.01in} 0.08 & \hspace{0.01in} \bf{0.99} & \hspace{0.01in} 0.96 & \hspace{0.01in} 0.38 & \hspace{0.01in} 0.02 & \hspace{0.01in} \bf{0.99} & \hspace{0.01in} 0.98 & \hspace{0.01in} 0.38 & \hspace{0.01in} 0.01 \\
D & \hspace{0.01in} \bf{0.77} & \hspace{0.01in} \bf{0.77} & \hspace{0.01in} 0.52 & \hspace{0.01in} 0.27 & \hspace{0.01in} \bf{0.99} & \hspace{0.01in} 0.96 & \hspace{0.01in} 0.52 & \hspace{0.01in} 0.28 & \hspace{0.01in} \bf{0.99} & \hspace{0.01in} \bf{0.99} & \hspace{0.01in} 0.52 & \hspace{0.01in} 0.32 & \hspace{0.01in} 0.99 & \hspace{0.01in} \bf{1.00} & \hspace{0.01in} 0.52 & \hspace{0.01in} 0.08 \\
E & \hspace{0.01in} \bf{1.00} & \hspace{0.01in} \bf{1.00} & \hspace{0.01in} 0.44 & \hspace{0.01in} 0.29 & \hspace{0.01in} \bf{1.00} & \hspace{0.01in} \bf{1.00} & \hspace{0.01in} 0.44 & \hspace{0.01in} 0.14 & \hspace{0.01in} \bf{1.00} & \hspace{0.01in} \bf{1.00} & \hspace{0.01in} 0.44 & \hspace{0.01in} 0.13 & \hspace{0.01in} \bf{1.00} & \hspace{0.01in} \bf{1.00} & \hspace{0.01in} 0.44 & \hspace{0.01in} 0.07 \\
F & \hspace{0.01in} \bf{0.88} & \hspace{0.01in} \bf{0.88} & \hspace{0.01in} 0.78 & \hspace{0.01in} 0.38 & \hspace{0.01in} \bf{0.99} & \hspace{0.01in} 0.96 & \hspace{0.01in} 0.78 & \hspace{0.01in} 0.39 & \hspace{0.01in} \bf{1.00} & \hspace{0.01in} 0.96 & \hspace{0.01in} 0.78 & \hspace{0.01in} 0.18 & \hspace{0.01in} \bf{1.00} & \hspace{0.01in} 0.96 & \hspace{0.01in} 0.78 & \hspace{0.01in} 0.27 \\
G & \hspace{0.01in} 0.68 & \hspace{0.01in} 0.68 & \hspace{0.01in} 0.65 & \hspace{0.01in} \bf{0.78} & \hspace{0.01in} \bf{1.00} & \hspace{0.01in} \bf{1.00} & \hspace{0.01in} 0.65 & \hspace{0.01in} 0.77 & \hspace{0.01in} \bf{1.00} & \hspace{0.01in} \bf{1.00} & \hspace{0.01in} 0.65 & \hspace{0.01in} 0.77 & \hspace{0.01in} \bf{1.00} & \hspace{0.01in} \bf{1.00} & \hspace{0.01in} 0.65 & \hspace{0.01in} 0.77 \\
H & \hspace{0.01in} \bf{0.30} & \hspace{0.01in} \bf{0.30} & \hspace{0.01in} 0.12 & \hspace{0.01in} 0.22 & \hspace{0.01in} 0.34 & \hspace{0.01in} \bf{0.34} & \hspace{0.01in} 0.13 & \hspace{0.01in} 0.23 & \hspace{0.01in} \bf{0.92} & \hspace{0.01in} \bf{0.92} & \hspace{0.01in} 0.13 & \hspace{0.01in} 0.22 & \hspace{0.01in} \bf{0.92} & \hspace{0.01in} \bf{0.92} & \hspace{0.01in} 0.13 & \hspace{0.01in} 0.37 \\
I & \hspace{0.01in} \bf{0.37} & \hspace{0.01in} \bf{0.37} & \hspace{0.01in} 0.13 & \hspace{0.01in} 0.00 & \hspace{0.01in} \bf{1.00} & \hspace{0.01in} 0.37 & \hspace{0.01in} 0.13 & \hspace{0.01in} 0.00 & \hspace{0.01in} \bf{1.00} & \hspace{0.01in} \bf{1.00} & \hspace{0.01in} 0.13 & \hspace{0.01in} 0.00 & \hspace{0.01in} \bf{1.00} & \hspace{0.01in} \bf{1.00} & \hspace{0.01in} 0.13 & \hspace{0.01in} 0.31 \\
J & \hspace{0.01in} 0.07 & \hspace{0.01in} 0.07 & \hspace{0.01in} 0.01 & \hspace{0.01in} \bf{1.00} & \hspace{0.01in} 0.41 & \hspace{0.01in} 0.07 & \hspace{0.01in} 0.01 & \hspace{0.01in} \bf{0.86} & \hspace{0.01in} \bf{0.80} & \hspace{0.01in} 0.09 & \hspace{0.01in} 0.04 & \hspace{0.01in} 0.75 & \hspace{0.01in} 0.80 & \hspace{0.01in} 0.09 & \hspace{0.01in} 0.04 & \hspace{0.01in} \bf{0.86} \\
K & \hspace{0.01in} \bf{0.28} & \hspace{0.01in} \bf{0.28} & \hspace{0.01in} 0.15 & \hspace{0.01in} 0.00 & \hspace{0.01in} \bf{0.99} & \hspace{0.01in} 0.27 & \hspace{0.01in} 0.15 & \hspace{0.01in} 0.00 & \hspace{0.01in} \bf{0.99} & \hspace{0.01in} \bf{0.99} & \hspace{0.01in} 0.15 & \hspace{0.01in} 0.00 & \hspace{0.01in} \bf{0.99} & \hspace{0.01in} \bf{0.99} & \hspace{0.01in} 0.15 & \hspace{0.01in} 0.00 \\
L & \hspace{0.01in} \bf{0.67} & \hspace{0.01in} \bf{0.67} & \hspace{0.01in} 0.07 & \hspace{0.01in} 0.37 & \hspace{0.01in} \bf{0.96} & \hspace{0.01in} 0.88 & \hspace{0.01in} 0.07 & \hspace{0.01in} 0.42 & \hspace{0.01in} \bf{0.96} & \hspace{0.01in} 0.88 & \hspace{0.01in} 0.07 & \hspace{0.01in} 0.08 & \hspace{0.01in} \bf{0.96} & \hspace{0.01in} 0.88 & \hspace{0.01in} 0.07 & \hspace{0.01in} 0.31 \\
M & \hspace{0.01in} \bf{0.92} & \hspace{0.01in} \bf{0.92} & \hspace{0.01in} 0.10 & \hspace{0.01in} 0.37 & \hspace{0.01in} \bf{0.99} & \hspace{0.01in} 0.92 & \hspace{0.01in} 0.10 & \hspace{0.01in} 0.46 & \hspace{0.01in} \bf{0.99} & \hspace{0.01in} 0.92 & \hspace{0.01in} 0.10 & \hspace{0.01in} 0.00 & \hspace{0.01in} \bf{0.99} & \hspace{0.01in} 0.92 & \hspace{0.01in} 0.10 & \hspace{0.01in} 0.18 \\
N & \hspace{0.01in} \bf{0.60} & \hspace{0.01in} \bf{0.60} & \hspace{0.01in} 0.02 & \hspace{0.01in} 0.00 & \hspace{0.01in} \bf{0.20} & \hspace{0.01in} 0.09 & \hspace{0.01in} 0.02 & \hspace{0.01in} 0.00 & \hspace{0.01in} 0.11 & \hspace{0.01in} \bf{0.21} & \hspace{0.01in} 0.02 & \hspace{0.01in} 0.00 & \hspace{0.01in} 0.18 & \hspace{0.01in} \bf{0.78} & \hspace{0.01in} 0.02 & \hspace{0.01in} 0.31 \\
O & \hspace{0.01in} \bf{0.11} & \hspace{0.01in} \bf{0.11} & \hspace{0.01in} 0.01 & \hspace{0.01in} 0.04 & \hspace{0.01in} \bf{0.50} & \hspace{0.01in} 0.17 & \hspace{0.01in} 0.01 & \hspace{0.01in} 0.21 & \hspace{0.01in} \bf{0.70} & \hspace{0.01in} 0.17 & \hspace{0.01in} 0.01 & \hspace{0.01in} 0.21 & \hspace{0.01in} \bf{1.00} & \hspace{0.01in} 0.21 & \hspace{0.01in} 0.01 & \hspace{0.01in} 0.00 \\
P & \hspace{0.01in} \bf{0.16} & \hspace{0.01in} \bf{0.16} & \hspace{0.01in} 0.04 & \hspace{0.01in} 0.04 & \hspace{0.01in} \bf{0.23} & \hspace{0.01in} 0.21 & \hspace{0.01in} 0.03 & \hspace{0.01in} 0.04 & \hspace{0.01in} \bf{0.82} & \hspace{0.01in} 0.21 & \hspace{0.01in} 0.03 & \hspace{0.01in} 0.14 & \hspace{0.01in} \bf{1.00} & \hspace{0.01in} 0.29 & \hspace{0.01in} 0.03 & \hspace{0.01in} 0.14 \\
Q & \hspace{0.01in} \bf{0.07} & \hspace{0.01in} \bf{0.07} & \hspace{0.01in} 0.02 & \hspace{0.01in} 0.02 & \hspace{0.01in} 0.16 & \hspace{0.01in} 0.12 & \hspace{0.01in} 0.01 & \hspace{0.01in} \bf{0.31} & \hspace{0.01in} \bf{0.92} & \hspace{0.01in} 0.12 & \hspace{0.01in} 0.01 & \hspace{0.01in} 0.18 & \hspace{0.01in} \bf{1.00} & \hspace{0.01in} 0.12 & \hspace{0.01in} 0.01 & \hspace{0.01in} 0.20 \\
\bottomrule
\end{tabular}
\label{tab:al_f1}
\end{table*}

\subsection{Synthesis Running Time}
\label{sec:runningtime}
\begin{table}[t]
\centering
\notsotiny
\caption{Running time (seconds) of synthesis (mean $\pm$ standard error) using binary search ($B$) and quantitative semantics ($Q$) running on CPU and GPU, with 25 steps of active learning.}
\begin{tabular}{lrrrrrrrr}
\toprule
\multicolumn{1}{c}{\multirow{2}{*}{\textbf{ID}}} & \multicolumn{4}{c}{\textbf{CPU}} & \multicolumn{4}{c}{\textbf{GPU}} \\
& \multicolumn{2}{c}{$B$} &
\multicolumn{2}{c}{$Q$} &
\multicolumn{2}{c}{$B$} &
\multicolumn{2}{c}{$Q$} \\
\midrule
A & $8,460$\hspace{0.03in} $\pm$ & $1,517$ \hspace{0.10in} & $3,343$\hspace{0.03in} $\pm$ & $202$ \hspace{0.10in} & $737$\hspace{0.03in} $\pm$ & $36$ \hspace{0.10in} & $174$\hspace{0.03in} $\pm$ & $14$ \hspace{0.10in}\\
B & $3,511$\hspace{0.03in} $\pm$ & $549$ \hspace{0.10in} & $2,291$\hspace{0.03in} $\pm$ & $237$ \hspace{0.10in} & $428$\hspace{0.03in} $\pm$ & $37$ \hspace{0.10in} & $110$\hspace{0.03in} $\pm$ & $9$ \hspace{0.10in}\\
C & $3,319$\hspace{0.03in} $\pm$ & $505$ \hspace{0.10in} & $2,007$\hspace{0.03in} $\pm$ & $359$ \hspace{0.10in} & $376$\hspace{0.03in} $\pm$ & $6$ \hspace{0.10in} & $113$\hspace{0.03in} $\pm$ & $9$ \hspace{0.10in}\\
D & $2,728$\hspace{0.03in} $\pm$ & $476$ \hspace{0.10in} & $2,714$\hspace{0.03in} $\pm$ & $334$ \hspace{0.10in} & $370$\hspace{0.03in} $\pm$ & $8$ \hspace{0.10in} & $119$\hspace{0.03in} $\pm$ & $2$ \hspace{0.10in}\\
E & $1,264$\hspace{0.03in} $\pm$ & $176$ \hspace{0.10in} & $599$\hspace{0.03in} $\pm$ & $54$ \hspace{0.10in} & $225$\hspace{0.03in} $\pm$ & $3$ \hspace{0.10in} & $50$\hspace{0.03in} $\pm$ & $1$ \hspace{0.10in}\\
F & $1,689$\hspace{0.03in} $\pm$ & $360$ \hspace{0.10in} & $748$\hspace{0.03in} $\pm$ & $81$ \hspace{0.10in} & $285$\hspace{0.03in} $\pm$ & $7$ \hspace{0.10in} & $60$\hspace{0.03in} $\pm$ & $1$ \hspace{0.10in}\\
G & $661$\hspace{0.03in} $\pm$ & $141$ \hspace{0.10in} & $133$\hspace{0.03in} $\pm$ & $23$ \hspace{0.10in} & $219$\hspace{0.03in} $\pm$ & $77$ \hspace{0.10in} & $30$\hspace{0.03in} $\pm$ & $1$ \hspace{0.10in}\\
H & $399$\hspace{0.03in} $\pm$ & $70$ \hspace{0.10in} & $147$\hspace{0.03in} $\pm$ & $9$ \hspace{0.10in} & $185$\hspace{0.03in} $\pm$ & $94$ \hspace{0.10in} & $32$\hspace{0.03in} $\pm$ & $17$ \hspace{0.10in}\\
I & $400$\hspace{0.03in} $\pm$ & $74$ \hspace{0.10in} & $84$\hspace{0.03in} $\pm$ & $13$ \hspace{0.10in} & $163$\hspace{0.03in} $\pm$ & $85$ \hspace{0.10in} & $23$\hspace{0.03in} $\pm$ & $12$ \hspace{0.10in}\\
J & $544$\hspace{0.03in} $\pm$ & $120$ \hspace{0.10in} & $173$\hspace{0.03in} $\pm$ & $5$ \hspace{0.10in} & $227$\hspace{0.03in} $\pm$ & $121$ \hspace{0.10in} & $36$\hspace{0.03in} $\pm$ & $19$ \hspace{0.10in}\\
K & $493$\hspace{0.03in} $\pm$ & $77$ \hspace{0.10in} & $125$\hspace{0.03in} $\pm$ & $25$ \hspace{0.10in} & $163$\hspace{0.03in} $\pm$ & $83$ \hspace{0.10in} & $30$\hspace{0.03in} $\pm$ & $16$ \hspace{0.10in}\\
L & $732$\hspace{0.03in} $\pm$ & $47$ \hspace{0.10in} & $286$\hspace{0.03in} $\pm$ & $73$ \hspace{0.10in} & $252$\hspace{0.03in} $\pm$ & $133$ \hspace{0.10in} & $57$\hspace{0.03in} $\pm$ & $29$ \hspace{0.10in}\\
M & $697$\hspace{0.03in} $\pm$ & $40$ \hspace{0.10in} & $253$\hspace{0.03in} $\pm$ & $49$ \hspace{0.10in} & $245$\hspace{0.03in} $\pm$ & $128$ \hspace{0.10in} & $56$\hspace{0.03in} $\pm$ & $30$ \hspace{0.10in}\\
N & $5,691$\hspace{0.03in} $\pm$ & $272$ \hspace{0.10in} & $977$\hspace{0.03in} $\pm$ & $176$ \hspace{0.10in} & $1,393$\hspace{0.03in} $\pm$ & $590$ \hspace{0.10in} & $264$\hspace{0.03in} $\pm$ & $136$ \hspace{0.10in}\\
O & $8,306$\hspace{0.03in} $\pm$ & $521$ \hspace{0.10in} & $2,314$\hspace{0.03in} $\pm$ & $476$ \hspace{0.10in} & $811$\hspace{0.03in} $\pm$ & $12$ \hspace{0.10in} & $127$\hspace{0.03in} $\pm$ & $2$ \hspace{0.10in}\\
P & $11,326$\hspace{0.03in} $\pm$ & $673$ \hspace{0.10in} & $4,198$\hspace{0.03in} $\pm$ & $1,333$ \hspace{0.10in} & $970$\hspace{0.03in} $\pm$ & $60$ \hspace{0.10in} & $167$\hspace{0.03in} $\pm$ & $8$ \hspace{0.10in}\\
Q & $12,430$\hspace{0.03in} $\pm$ & $962$ \hspace{0.10in} & $2,915$\hspace{0.03in} $\pm$ & $508$ \hspace{0.10in} & $1,141$\hspace{0.03in} $\pm$ & $101$ \hspace{0.10in} & $183$\hspace{0.03in} $\pm$ & $11$ \hspace{0.10in}\\
\bottomrule
\end{tabular}
\label{tab:runtime}
\end{table}
Next, we show that quantitative pruning and using a GPU each significantly reduce synthesis time, evaluating total running time for 25 steps of active learning.

\para{Ablations.}
We compare to two ablations: (i) using the binary search approach of~\cite{efficientenumerative} to find pruning pairs, rather than using our quantitative semantics, and (ii) evaluating the matrix semantics (\refappendix{sec:matrix}) on a CPU rather than a GPU.

\para{Results.}
In Table~\ref{tab:runtime}, we report the running time of our algorithms on a CPU ($2\times$ AMD EPYC 7402 24-Core) and a GPU ($1\times$ NVIDIA RTX A6000). For binary search, on average, the GPU is $7.6\times$ faster than the CPU. On a GPU, using the quantitative semantics rather than binary search offers another $5.0\times$ speed-up.

\section{Related Work}
\label{sec:related}
\para{Monotonicity for parameter pruning.}
We build on~\cite{learningmonotone} for our parameter pruning algorithm. Their approach has been applied to synthesizing STL formulas for sequence classification by first enumerating sketches and then using monotonicity to find parameters, similar to our binary search baseline~\cite{efficientenumerative}. We replace binary search with our novel strategy based on quantitative semantics, leading to
$5.0 \times$ speedup.
There is also work building on~\cite{learningmonotone} to create logically-relevant distance metrics between trajectories by taking the Hausdorff distance between parameter satisfaction regions (which they call ``validity domains''), with applications to clustering~\cite{monotoniclogical}. For logics like STL, our quantitative semantics could provide a speedup to their approach.

\para{Synthesis of temporal logic formulas.}
More broadly, there has been work synthesizing parameters in a variant of STL by discretizing the parameter space and then walking the satisfaction boundary~\cite{synthptstl}; in one dimension, their approach becomes binary search, inheriting its shortcomings.
There has been work on synthesizing STL formulas that are satisfied by a closed-loop control model~\cite{miningstl}, but they assume the ability to find counterexample traces for incorrect STL formulas, which is not applicable to our setting. 
Another approach is to synthesize parameters in STL formulas
using gradient-based optimization~\cite{jha2019} or stochastic optimization~\cite{kong2014}, but we found these methods to be ineffective in our setting, and they do not come with either soundness or completeness guarantees.
There is work using decision trees to synthesize STL formulas~\cite{kong2017,linard2020,bombara2021,aasi2022}, but these operate on a restricted subset of STL, namely Boolean combinations of a fixed set of template formulas. This restriction prevents these approaches from synthesizing temporal structure, which is a key component of the queries in our domains. Finally, there has been work on active learning of STL formulae using decision trees~\cite{linard2020}, but it assumes the ability to query for equivalence between a particular STL formula and the ground truth, which is not possible in our setting.

\para{Synthesizing constants.} There is work on synthesizing parameters of programs using counterexampled-guided inductive synthesis and different theory solvers, including Fourier–Motzkin variable elimination and an SMT solver~\cite{alessandro2018cegis}. Though our synthesis objective can be encoded in the theory of linear arithmetic, it is extremely large, and we have found such solvers to be ineffective in practice.

\para{Querying video data.}
There has been recent work on querying object detections and trajectories in video data~\cite{kang2019blazeit,fu2019rekall,kang2020jointly,kang2020model,bastani2020miris,bastani2020vaas,moll2020exsample,bastaniskyquery}.
The main difference is our focus on synthesis; in addition, these approaches focus on SQL-like operators such as select, inner-join, group-by, etc., over predefined predicates, which cannot capture compositions such as the sequencing and iteration operators in our language, which are necessary for identifying more complex behaviors.

\para{Neurosymbolic models.}
There has been recent work on leveraging program synthesis in the context of machine learning. For instance, there has been work on using programs to represent high-level structure in images~\cite{ellis2015unsupervised,ellis2018learning,tian2019learning,young2019learning,ellis2019write}, for reinforcement learning~\cite{verma2018programmatically,bastani2018verifiable,inala2019synthesizing,verma2019imitation}, and for querying websites~\cite{chen2021web}; in contrast, we use programs to classify trajectories. The most closely related work is on synthesizing functional programs operating over lists~\cite{valkov2018houdini,shah2020learning}. Our language includes key constructs not included in their languages. Most importantly, we include sequencing; in their functional language, such an operator would need to be represented as a nested series of if-then-else operators. In addition, their language does not support predicates that match subsequences; while such a predicate could be added, none of their operators can compose such predicates.

\para{Quantitative synthesis.}
There has been work on program synthesis with quantitative properties---e.g., on synthesis for producing optimized code~\cite{schkufza2013stochastic,phothilimthana2016scaling,jia2019optimizing}, for approximate computing~\cite{misailovic2014chisel,bornholt2016optimizing}, for probabilistic programming~\cite{nori2015efficient}, and for embedded control~\cite{chaudhuri2014bridging}. These approaches largely focus on search-based synthesis, either using constraint optimization~\cite{misailovic2014chisel}, continuous optimization~\cite{chaudhuri2014bridging}, enumerative search~\cite{bornholt2016optimizing,phothilimthana2016scaling},
or stochastic search~\cite{schkufza2013stochastic,nori2015efficient,jia2019optimizing}. While we leverage ideas from this literature, our quantitative semantics based pruning strategy is novel.

\para{Quantitative semantics.}
Our quantitative semantics is similar to the ``robustness degree''~\cite{fainekos2009robustness} of a temporal logic formula. The difference is that, by adjusting the denotations of the base predicates, our quantitative semantics gives a parameter on the satisfaction boundary.
More broadly, there has been work on quantitative semantics for temporal logic for robust constraint satisfaction~\cite{fainekos2009robustness,tabuada2016robust,deshmukh2017robust}, and to guide reinforcement learning~\cite{jothimurugan2019composable}. There has been work on quantitative regular expressions (QREs)~\cite{alur2017derivatives}, though in general, QREs cannot be efficiently evaluated due to their nondeterminism, and our language is restricted to ensure efficient computation. There has been work on synthesizing QREs for network traffic classification~\cite{shi2021}, using binary search to compute decision thresholds. Similarly, there has been work using the Viterbi semiring to obtain quantitative semantics for Datalog programs~\cite{si2019synthesizing}, which they use in conjunction with gradient descent to learn the rules of the Datalog program. In contrast, we use our quantitative semantics to efficiently prune the parameter search space in a provably correct way. Finally, there has been work on using GPUs to evaluate regular expressions~\cite{naghmouchi2010small}; however, they focus on regular expressions over strings.

\para{Query languages.} Our language is closely related to both signal temporal logic (STL)~\cite{stl} and Kleene algebras with tests (KAT)~\cite{kozen1997kleene}. In particular, it can straightforwardly be extended to subsume both (see \refappendix{sec:langexts} for details), and our pruning strategy applies to this extended language. The addition of Kleene star, required to subsume KAT, worsens the evaluation time. STL has been used to monitor safety requirements for autonomous vehicles~\cite{hekmatnejad2019}. Spatio-Temporal Perception Logic (SPTL) is an extension of STL to support spatial reasoning~\cite{hekmatnejad2022formalizing}. Many of its operators are monotone, and thus would benefit from our algorithm. Scenic~\cite{fremont2019scenic} is a DSL for creating static and dynamic driving scenes, but its focus is on generating scenes rather than querying for behaviors.

\section{Conclusion}

We have proposed a novel framework called \toolname for synthesizing queries over video trajectory data.
Our language is similar to KAT and STL, but supports conjunction and sequencing over multi-step predicates.
Given only a few examples, \toolname efficiently synthesizes trajectory queries consistent with those examples.
A key contribution of our approach is the use of a quantitative semantics to prune the parameter search space, yielding a $5.0 \times$ speedup over the state-of-the-art.
In our evaluation, we demonstrate that \toolname effectively synthesizes queries to identify interesting driving behaviors, and that our optimizations dramatically reduce synthesis time.

\subsubsection*{Acknowledgements}
We thank the anonymous reviewers for their helpful feedback. This work was supported in part by NSF Award CCF-1910769, NSF Award CCF-1917852, and ARO Award W911NF-20-1-0080.

\subsubsection*{Data Availability Statement}
An artifact is available~\cite{artifact} with our implementation, which reproduces our experimental results (Table~\ref{tab:al_f1} and Table~\ref{tab:runtime}) and may be useful for performing synthesis on other trajectory datasets or implementing our algorithm for other DSLs.

%% file: appendix.tex
\section{Query Language Extensions}
\label{sec:langexts}

\subsection{Evaluating Quantitative Semantics with Matrices}
\label{sec:matrix}
The sequencing operator requires considering arbitrary splits of a trajectory, which can be accomplished with dynamic programming. However, this evaluation is highly data-parallel, something GPUs excel at. Because matrix operations are a natural fit for GPU programming, we devise a matrix semantics for our language, achieving a large additional speedup.

Instead of just computing the quantitative semantics for a trajectory, the idea is to simultaneously evaluate the semantics for every subsequence, i.e. $\denmat{Q}_{v,u}(z)$ is an $(n+1) \times (n+1)$ matrix where the $i,j$-th entry gives the quantitative semantics for the subsequence from $i$ to $j$---i.e. we define the semantics such that Theorem~\ref{thm:matrixsound} below holds.

\begin{figure}[t]
\begin{align*}
\denmat{\varphi_{??}}(z) &\coloneqq
\begin{bmatrix}
\iota_{\varphi}(z_{0:0}) & \iota_{\varphi}(z_{0:1}) & \cdots & \iota_{\varphi}(z_{0:n}) \\
-\infty & \iota_{\varphi}(z_{1:1}) & \cdots & \iota_{\varphi}(z_{1:n}) \\
\vdots & \vdots & \ddots & \vdots \\
-\infty & -\infty & \cdots & \iota_{\varphi}(z_{n:n}) \\
\end{bmatrix} \\
\denmat{Q_1 \land Q_2}(z) &\coloneqq \min\{\denmat{Q_1}(z), \denmat{Q_2}(z)\} \\
\denmat{Q_1 \seq Q_2}(z) &\coloneqq \denmat{Q_1}(z) \cdot \denmat{Q_2}(z)
\end{align*}
\caption{The matrix semantics of our language; $z\in\mathcal{Z}$ is a trajectory of length $n$; $\min$ is elementwise; $\cdot$ is the matrix product with $\max$ instead of addition and $\min$ instead of multiplication. We omit the case $\denmat{\varphi_{\theta}}(z)$, which is the same as $\denmat{\varphi_{??}}(z)$, except $\iota_{\varphi}(z_{i:j})$ is replaced with $\denq{\varphi}(z_{i:j})$.}
\label{fig:matrixsemantics}
\end{figure}

We define the matrix semantics in Figure~\ref{fig:matrixsemantics}. The base case $\varphi_{??}$ evaluates $\iota_\varphi$ on every subsequence, conjunction is element-wise minimum, and sequencing is matrix multiplication in the max-min semiring. Na\"{i}vely, evaluating this with standard matrix multiplication techniques would take $O(n^3)$ time and memory. However, since we are ultimately interested in $$\denq{Q}_{v,u}(z) = \big(\denmat{Q}_{v,u}(z)\big)_{0,n} = \begin{bmatrix}\infty, -\infty,\ldots,-\infty\end{bmatrix} \cdot \big(\denmat{Q}_{v,u}(z)\big) \cdot \begin{bmatrix}-\infty,\ldots,-\infty,\infty\end{bmatrix}^T$$ we can reassociate the matrix-matrix products to be matrix-vector products. This uses only $O(n^2)$ time and memory, the same as dynamic programming.

\begin{lemma}
\label{lem:uppertriangular}
For any $Q$, vector $v \in \mathbb{R}^d$, positive vector $u \in \mathbb{R}_{>0}^d$, and trajectory $z$ of length $n$, $\denmat{Q}_{v,u}(z)$ is upper-triangular---i.e., $(\denmat{Q}_{v,u}(z))_{i,j}=-\infty$ for all $i>j$.
\end{lemma}
\begin{proof}
By induction on $Q$. The base cases are upper-triangular by definition, and both pointwise minimum and max-min matrix multiplication preserve upper-triangularity.
\end{proof}

\begin{theorem}
\label{thm:matrixsound}
For any query $Q$, vector $v \in \mathbb{R}^d$, positive vector $u \in \mathbb{R}_{>0}^d$, trajectory $z$ of length $n$, and indices $i \leq j$,  $(\denmat{Q}_{v,u}(z))_{i, j} = \denq{Q}_{v,u}(z_{i:j})$. In particular, $\denq{Q}_{v,u}(z) = (\denmat{Q}_{v,u}(z))_{0, n}$.
\end{theorem}
\begin{proof}
We prove by structural induction on $Q$. The base cases $Q=\varphi$ and $Q=\varphi_{??}$ follow by the definitions of $\denq{\varphi_{??}}(z_{i:j})$ and $(\denmat{\varphi_{??}}(z))_{i,j}$.

For the case $Q=Q_1\seq Q_2$, we have
\begin{align*}
\denmat{Q_1\seq Q_2}_{v,u}(z)_{i,j} &=\max_{1 \leq k \leq n}\min\{\denmat{Q_1}_{v,u}(z)_{i,k},\denmat{Q_2}_{v,u}(z)_{k,j}\} \\
\intertext{but since $\denmat{Q_1}_{v,u}(z)$ and $\denmat{Q_2}_{v,u}(z)$ are upper-triangular, we have}
&=\max_{i \leq k \leq j}\min\{\denmat{Q_1}_{v,u}(z)_{i,k},\denmat{Q_2}_{v,u}(z)_{k,j}\} \\
\intertext{and by the inductive hypothesis we have}
&=\max_{k\in\{i,i+1,...,j\}}\min\{\denq{Q_1}_{v,u}(z_{i:k}),\denq{Q_2}_{v,u}(z_{k:j})\} \\
&=\denq{Q_1\seq Q_2}_{v,u}(z_{i:j}).
\end{align*}

Finally, for the case $Q=Q_1\wedge Q_2$, we have
\begin{align*}
\denmat{Q_1\wedge Q_2}_{v,u}(z)_{i,j}
&=\min\{\denmat{Q_1}_{v,u}(z)_{i,j},\denmat{Q_2}_{v,u}(z)_{i,j}\} \\
&=\min\{\denq{Q_1}_{v,u}(z_{i:j}),\denq{Q_2}_{v,u}(z_{i:j})\} \\
&=\denq{Q_1\wedge Q_2}_{v,u}(z_{i:j}).
\end{align*}
\end{proof}

\subsection{Additional Operators}
\label{sec:addops}
We can add several additional operators to our language, including disjunction and Kleene star, which lets us subsume regular expressions, and negation and another operator, which lets us subsume signal temporal logic.
\begin{align*}
Q ::= \ldots \mid Q \lor Q \mid Q^* \mid \neg Q \mid Q^{\dashv [a,b]}
\end{align*}

The semantics of disjunction is straightfoward, being the same as conjunction but with $\max$ rather than $\min$.
\begin{align*}
\denmat{Q_1 \lor Q_2}_{v,u}(z) &\coloneqq \max\{\denmat{Q_1}_{v,u}(z),\denmat{Q_2}_{v,u}(z)\}
\end{align*}

Kleene star is a disjunction over repeated sequencings. The evaluation of Kleene star has a worse asymptotic complexity as a function of $n$ than the other operators.
\begin{align*}
\denmat{Q^*}_{v,u}(z) &\coloneqq \max_{0 \leq k \leq n} (\denmat{Q}_{v,u}(z))^k
\end{align*}

For negation, note that
\begin{align*}
\neg \den{\varphi_{\theta}}(z) = \neg \mathbbm{1}(\iota_{\varphi}(z)\ge\theta) &= \mathbbm{1}(\iota_{\varphi}(z) < \theta) = \mathbbm{1}(-\iota_{\varphi}(z) > -\theta)
\end{align*}
Thus if we introduce, for every atomic formula $\varphi$, a negated version $\bar{\varphi}$ where $\iota_{\bar{\varphi}}(z) \coloneqq -\iota_{\varphi}(z)$, then we have
\begin{align*}
\mathbbm{1}(-\iota_{\varphi}(z) > -\theta) &\approx \mathbbm{1}(-\iota_{\varphi}(z) \geq -\theta) = \den{\bar{\varphi}_{-\theta}}(z).
\end{align*}
Here we ignore technicalities regarding points lying on the satisfaction region boundary, and whether they are part of the satisfaction region itself (i.e. when $\iota_{\varphi}(z) = \theta$). We define $\bar{Q}$ to be the same formula as $Q$, except with each primitive $\varphi$ replaced with $\bar{\varphi}$. Then we can define the semantics for negation as:
\begin{align*}
\left(\denmat{\neg Q}_{v,u}(z)\right)_{i,j} &\coloneqq \begin{cases}
        -\left(\denmat{\bar{Q}}_{v,u}(z)\right)_{i,j} & i \leq j \\
        -\infty & \\
    \end{cases}
\end{align*}

The operator ${(\cdot)}^{\dashv [a,b]}$ exists to support the ``until'' operator of STL. To define its semantics, we use two helper functions. The first takes a vector, representing the value of a signal at each step in time, and transforms it into a matrix, representing the minimum value of the signal for each interval with between $a$ and $b$ time-steps:
\begin{align*}
    (\cdot)^{[a,b]} &: \mathbb{R}^n \to \mathbb{R}^{n \times n} \\
    (v^{[a,b]})_{i, j} &\coloneqq \begin{cases}
        \min\limits_{i \leq k \leq j} v_k & a \leq j - i \leq b\\
        -\infty & \\
    \end{cases}
\end{align*}
The second takes a matrix of values for each time interval and returns a vector of, for each time step $i$, the value of the interval from $i$ to the end of the sequence:
\begin{align*}
    (\cdot)^\dashv &: \mathbb{R}^{n \times n} \to \mathbb{R}^{n} \\
    (M^\dashv)_{i} &\coloneqq M_{i, n}
\end{align*}
Now we can define the operator's semantics:
\begin{align*}
\denmat{Q^{\dashv [a,b]}}_{v,u}(z) &\coloneqq (\denmat{Q}_{v,u}(z)^\dashv)^{[a,b]}
\end{align*}

Finally, we can define the translation $\transstl{\cdot}$ from STL formulas. It is straightforward, except the ``until'' operator, which decomposes into sequencing and the $(\cdot)^{\dashv [a,b]}$ operator:
\begin{align*}
\transstl{\varphi} &\coloneqq \varphi \\
\transstl{\neg Q_1} &\coloneqq {\neg \transstl{Q}} \\
\transstl{Q_1 \land Q_2} &\coloneqq \transstl{Q_1} \land \transstl{Q_2} \\
\transstl{Q_1 \mathbf{U}_{[a,b]} Q_2} &\coloneqq \transstl{Q_1}^{\dashv [a,b]} \seq \transstl{Q_2}
\end{align*}

\section{Predicates and Ground-Truth Queries for Datasets}
\label{sec:predandquery}
On all datasets we implicitly include $\predbase{\textsf{Any}}$ and $\predbase{\textsf{None}}$ predicates, which match all trajectories of any length and no trajectories, respectively.

We have three kinds of queries in the YTStreams dataset: one-object queries in the Shibuya video, one-object queries in the Warsaw video, and two-object queries in the Warsaw video. The number of lanes varies between Shibuya and Warsaw, and for one-object queries, certain predicates are not applicable. Including these multiplicities, Shibuya has $9$ predicates, one-object Warsaw has $10$, and two-object Warsaw has $19$. Including multiplicities, MABe22 has $16$ predicates. The maritime surveillance task has $8$ predicates.

The way queries are evaluated varies slightly between the maritime surveillance task and other tasks. In the maritime surveillance task, a trajectory $z$ is considered to match a query $Q$ if $\den{Q}(z) = 1$. The other tasks are drawn from segments of video, and thus may not represent whole trajectories (a trajectory in our dataset may appear in the middle of the road rather than entering from the side). In these cases, we consider a trajectory positive if there is some subset trajectory that matches, so $\den{Q}(z_{i:j}) = 1$ for some $i \leq j$. This is equivalent to sequencing the beginning and end of $Q$ with $\predbase{\textsf{Any}}$, so checking that $\den{\predbase{\textsf{Any}} \seq Q \seq \predbase{\textsf{Any}}}(z) = 1$.

\begin{table*}[t]
\centering
\scriptsize
\caption{The predicates used for the MABe22 dataset.}
\begin{tabular}{>{\centering\arraybackslash}p{0.25\textwidth}p{0.75\textwidth}}
\toprule
{\bf Predicate} & \multicolumn{1}{c}{{\bf Description}} \\
\midrule
$\textsf{DistanceLt}_{\theta}$ & Whether, at every time-step in the interval, the distance between the two objects is less than $\theta$. \vspace{5pt} \\
$\textsf{DistanceGt}_{\theta}$ & Whether, at every time-step in the interval, the distance between the two objects is greater than $\theta$. \vspace{5pt} \\
$\textsf{Running}(A)$ & Whether, at every time-step in the interval, the speed of mouse A was at least 15 cm/sec. \vspace{5pt} \\
$\textsf{MovingSlowly}(A)$ & Whether, at every time-step in the interval, the speed of mouse A was at most 2 cm/sec. \vspace{5pt} \\
$\textsf{DurationShort}$ & Whether the interval spans at most 5 seconds. \vspace{5pt} \\
$\textsf{DurationNotShort}$ & Whether the interval spans at least 5 seconds. \vspace{5pt} \\
$\textsf{DurationLong}$ & Whether the interval spans at least 10 seconds. \vspace{5pt} \\
$\textsf{NoseNoseContact(A,B)}$ & Whether, at every time-step in the interval, the distance between mouse A’s nose and mouse B’s nose is at most 1.5 cm. \vspace{5pt} \\
$\textsf{NoseGenitalContact(A,B)}$ & Whether, at every time-step in the interval, the distance between mouse A’s nose and mouse B’s tail base is at most 1.5 cm. \vspace{5pt} \\
$\textsf{NoseEarContact(A,B)}$ & Whether, at every time-step in the interval, the distance between mouse A’s nose and one of mouse B’s ears is at most 1.5 cm. \vspace{5pt} \\
$\textsf{MovingToward(A,B)}$ & Whether, at every time-step in the interval, the angle of A’s movement is close to the angle between A and B. \vspace{5pt} \\
$\textsf{GazingAt(A,B)}$ & Whether, at every time-step in the interval, the A’s head is pointing toward B. \vspace{5pt} \\
\bottomrule
\end{tabular}
\label{tab:predicates:mabe}
\end{table*}

\begin{table*}[tp]
    \centering
    \scriptsize
    \caption{The predicates used for the maritime surveillance task.}
    \begin{tabular}{>{\centering\arraybackslash}p{0.2\textwidth}p{0.75\textwidth}}
    \toprule
    {\bf Predicate} & \multicolumn{1}{c}{{\bf Description}} \\
    \midrule
    $\textsf{DurationGt}_{\theta}$ & Whether the interval spans at least $\theta$ seconds. \\
    $\textsf{DurationLt}_{\theta}$ & Whether the interval spans at most $\theta$ seconds. \\
    $\textsf{XPosGt}_{\theta}$ & Whether, at every time-step in the interval, the X position is at least $\theta$. \\
    $\textsf{XPosLt}_{\theta}$ & Whether, at every time-step in the interval, the X position is at most $\theta$. \\
    $\textsf{YPosGt}_{\theta}$ & Whether, at every time-step in the interval, the Y position is at least $\theta$. \\
    $\textsf{YPosLt}_{\theta}$ & Whether, at every time-step in the interval, the Y position is at most $\theta$. \\
    \bottomrule
    \end{tabular}
    \label{tab:predicates:ms}
\end{table*}
    
\begin{table*}[tp]
\scriptsize
\centering
\caption{Ground-truth queries for the MABe22 tasks. ``IDs'' indicates which queries evaluated in other tables are instances of these general queries. These are based on the behavior annotations from the MABe22 dataset~\cite{mabe22}. (See Table 2 in their Appendix B.1.)}
\begin{tabular}{>{\centering\arraybackslash}p{0.1\textwidth}>{}p{0.88\textwidth}}
\toprule
\textbf{IDs} & \textbf{Query} \\
\midrule
A &
\textbf{Approach}: Mice move from at least 5 cm apart to less than 1 cm apart at closest point, over a period of at least 10 seconds at a maximum speed of 2 cm/sec.

\vspace{0.5em}
\centerline{$\predbase{\textsf{DistanceGt}}_{??} \seq \left(\predbase{\textsf{MoveSlowly}}(A) \land \predbase{\textsf{DurationLong}} \right) \seq \predbase{\textsf{DistanceLt}}_{??}$} \\
B &
\textbf{Oral-oral contact}: Noses of mice are less than 1.5 cm apart. Must occur less than 5 seconds after an approach.

\vspace{0.5em}
\centerline{$\predbase{\textsf{DistanceGt}}_{??} \seq \left(\predbase{\textsf{MoveSlowly}}(A) \land \predbase{\textsf{DurationLong}} \right) \seq \qquad\ $} \centerline{$\qquad\  \predbase{\textsf{DistanceLt}}_{??} \seq \predbase{\textsf{DurationShort}} \seq \predbase{\textsf{NoseNoseContact}(A,B)}$} \\
C &
\textbf{Oral-genital contact}: Nose and tail base of mice are less than 1.5 cm apart. Must occur less than 5 seconds after an approach.

\vspace{0.5em}
\centerline{$\predbase{\textsf{DistanceGt}}_{??} \seq \left(\predbase{\textsf{MoveSlowly}}(A) \land \predbase{\textsf{DurationLong}} \right) \seq \qquad\quad$} \centerline{$\qquad\quad \predbase{\textsf{DistanceLt}}_{??} \seq \predbase{\textsf{DurationShort}} \seq \predbase{\textsf{NoseGenitalContact}(A,B)}$} \\
D &
\textbf{Oral-ear contact}: Nose and ear of mice are less than 1.5 cm apart. Must occur less than 5 seconds after an approach.

\vspace{0.5em}
\centerline{$\predbase{\textsf{DistanceGt}}_{??} \seq \left(\predbase{\textsf{MoveSlowly}}(A) \land \predbase{\textsf{DurationLong}} \right) \seq \qquad\ $} \centerline{$\qquad\  \predbase{\textsf{DistanceLt}}_{??} \seq \predbase{\textsf{DurationShort}} \seq \predbase{\textsf{NoseEarContact}(A,B)}$} \\
E &
\textbf{Chase}: Mice are moving above 15 cm/sec, with closest points less than 5 cm apart, and angular deviation between mice is less than 30 degrees.

\vspace{0.5em}
\centerline{$\predbase{\textsf{Running}(A)} \land \predbase{\textsf{Running}(B)} \land \predbase{\textsf{DistanceLt}}_{??} \land \predbase{\textsf{MovingToward}(A,B)}$} \\
F &
\textbf{Watching}: Mice are more than 5 cm apart by less than 20 cm apart, and gaze offset of one mouse is less than 15 degrees from body of other mouse, for a minimum duration of 3 seconds.

\vspace{0.5em}
\centerline{$\predbase{\textsf{DistanceLt}}_{??} \land \predbase{\textsf{DistanceGt}}_{??} \land \predbase{\textsf{GazingAt}(A,B)} \land \predbase{\textsf{DurationNotShort}}$} \\
\bottomrule
\end{tabular}
\label{tab:queries:mabe}
\end{table*}

See Table~\ref{tab:predicates:yts} for the YTStreams predicates and Table~\ref{tab:queries:yts} for the YTStreams ground-truth queries. See Table~\ref{tab:predicates:mabe} for the MABe22 predicates and Table~\ref{tab:queries:mabe} for the MABe22 ground-truth queries.  See Table~\ref{tab:predicates:ms} for the maritime surveillance predicates.

\section{Neural Network Baselines}
\label{sec:nnbaselines}
Input features vary depending on the dataset: for one-object Shibuya, one-object Warsaw, and maritime surveillance, states consist of $x$ position, $y$ position, and time (state dimension $3$). For two-object Warsaw, states consist of the $x$ and $y$ positions of each object, time, and a mask $1$ or $0$ for each object to indicate its presence in that state. For MABe22, there are two mice, each of which has 12 tracked keypoints, each of which is an $x$, $y$ pair.

During active learning, we label the unlabeled trajectory with the highest predicted probability of being positive.
\subsubsection{Transformer}
We train a transformer model to directly predict class labels using the cross-entropy loss. We adapted the trajectory-forecasting model from~\cite{giuliari2020transformer} but modify the final layer to predict class logits rather than future positions.

\subsubsection{LSTM}
First we train a decoder-only LSTM (adapting the code from Argoverse v1~\cite{argoverse}) on the task of predicting the next state of each trajectory of a dataset. The LSTM  embedding size is twice the input size, and the hidden state size is four times the input size. We train it for 1000 epochs. The trajectory embeddings used for logistic regression are the final hidden states from the LSTM.

\section{Proofs}

\subsection{Proof of Theorem~\ref{thm:algcomplete}}
\label{sec:proof:algcomplete}
Because we pop boxes in first-in-first-out order, we do a breadth-first search. Thus as $N\to\infty$, also the depth in the search tree goes to infinity. So it suffices to show that, at each additional depth, we prune at least a constant fraction of the space. In particular, we prune at least $1/9$ of the remaining space at each additional depth, and $(1 - 1/9)^d \to 0$ as $d\to\infty$.

When doing pruning on a box $b$, regardless of whether the pruning pair for $b$ is consistent or inconsistent, we always prune $b_{\operatorname{lower}}$, $b_{\operatorname{upper}}$, and $b_{\operatorname{center}}$. These three boxes lie along the centerline of $b$, and moreover their centerlines are a partition of the centerline of $b$. So at least one of them, call it $b'$, contains $\geq 1/3$ of the centerline of $b$. Thus $b'$ spans at least $1/3$ of the width of $b$ and $1/3$ of the height, or at least $1/9$ of the area. Since at each depth in the search tree, every box is pruned by at least $1/9$, then the overall remaining space is also pruned by at least $1/9$.
\subsection{Proof of Theorem~\ref{thm:quantcorrect}}
\label{sec:proof:quantcorrect}
Fix a parameter $v \in \mathbb{R}^d$, and a positive vector $u \in \mathbb{R}_{>0}^d$. Let $t_{Q,z}^* \coloneq \denq{Q}_{v,u}(z)$ and $f_{Q,z}(t) \coloneq \den{Q(t \cdot u + v)}(z)$.
Conceptually, $f_{Q,z}$ considers the line segment in parameter space from $v$ ($t = 0$) to $v + u$ ($t = 1$), giving whether the point $t$ along that line segment is in the satisfaction region for $Q$ and $z$.

We say that $t$ is a boundary value of an antitone function $f : \bar{\mathbb{R}} \to \mathbb{B}$ if for any $t' > t$, $f(t') = 0$ and either $t = -\infty$ or $f(t) = 1$. Note that if $t$ is a boundary value of $f_{Q,z}$, then $t \cdot u + v$ is a boundary parameter for $Q$ and $z$. Thus it suffices to show that $t_{Q,z}^*$ is a boundary value of $f_{Q,z}$.

First, we show that if $t_1^*$ and $t_2^*$ are boundary values of $f_1$ and $f_2$, respectively, then $t^* = \min\{t_1^*, t_2^*\}$ is a boundary value of $f(t) = f_1(t) \land f_2(t)$---i.e. the boundary value of a conjunction is the minimum of the boundary values. Suppose $t^* = t_1^* \leq t_2^*$ (the other case is symmetric). Thus, for any $t' > t^*$, $t' > t_1^*$, so $f_1(t') = 0$, and so $f(t') = 0$. Also, either $t_Q^* = -\infty$, in which case we are done, or $f_1(t^*) = f_1(t_1^*) = 1$. In the latter case, we also have $f_2(t_2^*) = 1$, and because $f_2$ is antitone, $f_2(t_Q^*) = 1$. Thus $f(t^*) = f_1(t^*) \land f_2(t^*) = 1$. A similar argument shows that  $\max\{t_1^*, t_2^*\}$ is a boundary value of $f(t) = f_1(t) \lor f_2(t)$---i.e. the boundary value of a disjunction is the maximum of the boundary values.

Now we show by induction on $Q$ that $t_{Q,z}^*$ is a boundary value of $f_{Q,z}$ for any $z$:

For $Q = \varphi_{??i}$, by definition $t_{Q,z}^* = \frac{\iota_{\varphi}(z) - v_i}{u_i}$ and $f_{Q,z}(t) = \mathbbm{1}(\iota_{\varphi}(z) \ge t \cdot u_i + v_i)$. Thus $f_{Q,z}(t) = \mathbbm{1}(\frac{\iota_{\varphi}(z) -v_1}{u_i} \ge t)$, so $t_{Q,z}^*$ is the threshold value of $f_{Q,z}$.

For $Q = \varphi$, first suppose $\textsf{sat}_{\varphi}(z) = 1$, so by definition $t_{Q,z}^* = \infty$ and $f_{Q,z}(t) = 1$. Since there is no $t' > t_{Q,z}^*$, $t_{Q,z}^*$ is a boundary value of $f_{Q,z}$. Now suppose $\textsf{sat}_{\varphi}(z) = 0$, so by definition $t_{Q,z}^* = -\infty$ and $f_{Q,z}(t) = 0$. Since for any $t' > t_{Q,z}^*$, $f_{Q,z}(t') = 0$, $t_{Q,z}^*$ is a boundary value of $f_{Q,z}$.

For $Q = Q_1 \land Q_2$, $t_{Q,z}^* = \min\{t_{Q_1,z}^*, t_{Q_2,z}^*\}$ and $f_{Q,z}(t) = f_{Q_1,z}(t) \land f_{Q_2,z}(t)$. By the inductive hypothesis, $t_{Q_1,z}^*$ is a boundary value of $f_{Q_1,z}$ and $t_{Q_2,z}^*$ is a boundary value of $f_{Q_2,z}$, and so by our earlier result, the boundary value of $f_{Q,z}$ (a conjunction) is $t_{Q,z}^*$ (a minimum).

For $Q = Q_1 \seq Q_2$, $t_{Q,z}^* = \max_{0 \leq k \leq n}\min\{t_{Q_1,z_{0:k}}^*,t_{Q_2,z_{k:n}}^*\}$ and $f_{Q,z}(t) = \bigvee_{k=0}^n ~ f_{Q_1, z_{0:k}}(t) \land f_{Q_2,z_{k:n}}(t)$. As before, we apply the inductive hypothesis, and then we apply our earlier result about conjunctions and minima, and finally our result about disjunctions and maxima, to conclude that the boundary value of $f_{Q,z}$ is $t_{Q,z}^*$.